%% file: main.tex
\documentclass[lean]{AFM}

\input{preamble}

\usepackage{tikz}
\usetikzlibrary{arrows, arrows.meta}
\definecolor{ghpurple}{rgb}{0.4, 0.22, 0.73}

\newcommand{\changeurlcolor}[1]{\hypersetup{urlcolor=#1}}   
\newcommand{\lsthref}[2]{\changeurlcolor{ghpurple}{\ttfamily\href{#1}{#2}}}

\title[Formalising the local compactness of the adele ring]{Formalising the local compactness of the adele ring} 

\author[Salvatore Mercuri]{Salvatore Mercuri}
\authorinfo[Salvatore Mercuri]{The University of Edinburgh, UK}{smercuri@ed.ac.uk}

\VOLUME{1}
\YEAR{2025}
\NUMBER{3}
\firstpage{133}
\DOI{https://doi.org/10.46298/afm.14840} 
\receiveddate{November 25, 2024}
\finaldate{April 29, 2025}
\accepteddate{July 2, 2025}

\begin{abstract}
The adele ring of a number field is a central object in modern number theory. 
Its status as a locally compact topological ring is one of the key reasons why. 
We describe a formal proof that the adele ring of a number field is locally compact implemented in the Lean 4 theorem prover. 
Our work includes the formalisations of new types, including the completion of a number field at an infinite place, the infinite adele ring and the finite $S$-adele ring, as well as formal proofs that completions of a number field are locally compact and that their rings of integers at finite places are compact. 
\end{abstract}

\begin{document}

\input{sections/introduction.tex}
\input{sections/preliminaries.tex}
\input{sections/absuniformspace.tex}
\input{sections/adelering.tex}
\input{sections/locallycompact.tex}
\input{sections/discussion.tex}

\section*{Acknowledgments}

I would like to thank Mar\'{i}a In\'{e}s de Frutos-Fern\'{a}ndez for their advice on formalising the infinite adele ring of a number field, and the \mathlib community for their constructive feedback on the code. I would also like to thank the anonymous reviewers for their helpful suggestions that have improved this paper.

%%
%% Bibliography
%%

\printbibliography

\end{document}

%% file: preamble.tex
\usepackage{amsmath, amsfonts, amsthm, framed}
\usepackage{xspace}
\usepackage{newunicodechar}
\usepackage{enumitem}
\usepackage{fontawesome}
\usepackage{soul}
\usepackage{cancel}
\usepackage{xspace}
\usepackage{eucal}
\usepackage{colonequals}
\definecolor{darkergreen}{rgb}{0.0, 0.5, 0.0}
\definecolor{Blue}{RGB}{0,0,148}

\definecolor{keywordcolor}{rgb}{0.7, 0.1, 0.1}   % red
\definecolor{commentcolor}{rgb}{0.4, 0.4, 0.4}   % grey
\definecolor{symbolcolor}{rgb}{0, 0, 0.8}    % blue
\definecolor{tacticcolor}{rgb}{0, 0, 0.8}    % blue
\definecolor{sortcolor}{rgb}{0.1, 0.5, 0.1}      % green
\theoremstyle{plain}
\theoremstyle{definition}
\DeclareMathOperator{\Spec}{Spec}

\newcommand*{\mathlib}{\textsc{mathlib}\xspace} %% Temporarily set this way to catch forgotten instances

%% file: sections/introduction.tex
\section{Introduction}
\label{sec:intro}

In number theory, looking locally can give a lot of information about the global picture with which we are primarily concerned. 
For example, the equation $x^2 + 3y^2 = 2024$ has no rational solution, because any such solution would give a local solution modulo $3$ of the equation $x^2 \equiv 2 \pmod{3}$, yet $2$ is not a square modulo $3$. 
More generally, Hasse's \emph{local-global principle} asserts that, under certain conditions, a global rational solution of an equation exists if and only if there is a real solution and local solutions for every prime $p$. 

The local study of the field $\mathbb{Q}$ of rational numbers at a prime $p$ is systematised in the field $\mathbb{Q}_p$ of $p$-adic numbers. 
Just as the real numbers $\mathbb{R}$ are the completion of $\mathbb{Q}$ with respect to the usual absolute value, the field $\mathbb{Q}_p$ is the completion of $\mathbb{Q}$ with respect to the $p$-adic absolute value. 
In line with the local-global principle, we might intuitively unify all real and local information into the product $\mathbb{R} \times \prod_p \mathbb{Q}_p$, where $p$ ranges over all primes. 
The rational adele ring $\mathbb{A}_{\mathbb{Q}}$ is a subring of this product defined by a restriction on the infinite product over $p$ which we detail later. 
This particular restriction makes the adele ring a \emph{locally compact} topological ring. 

In addition to encoding local information, the adele ring's local compactness has led to its widespread use in modern number theory. 
One particular example is the ability to do Fourier analysis on the adele ring, which Tate \cite{tate50} used to reformulate Dirichlet $L$-functions and their functional equations on adeles. 
These functional equations have normalisation factors which, in the classical setting, appeared artificially but in the adelic setting appear naturally from the real component of the rational adele ring. 
Tate's thesis is really the theory of automorphic forms on $GL(1)$ and a precursor to the Langlands program where the adele ring analogously plays a foundational role.

Adele rings of general number fields can be defined similarly to the rational adele ring. 
In this paper, we formalise the proof that the adele ring of a number field is a locally compact topological ring in the Lean theorem prover, \lstinline{v4.10.0}, with the use of Lean's mathematics library \mathlib \cite{mathlib} at commit \href{https://github.com/leanprover-community/mathlib4/tree/caac5b13fb72ba0c5d0b35a0067de108db65e964}{\small\texttt{caac5b1}}. 
All hyperlinks to \mathlib in this paper link to this commit; objects and results introduced after this commit are referenced via their introductory PRs.
The reference project code is available on GitHub\footnote{\href{https://github.com/smmercuri/adele-ring_locally-compact/tree/journal}{https://github.com/smmercuri/adele-ring\_locally-compact/tree/journal}}.
To our knowledge, this is the first formalised proof of the local compactness of the adele ring in any interactive theorem prover, marking a milestone along the long-term path of formalising Fermat's Last Theorem, Tate's thesis and the Langlands program within \mathlib. 

Prior work \cite{defrutosfernandez22} formalised the adele ring of a global field in Lean.
Developments in \mathlib since then enable us to formalise the completion of a number field at its infinite places, which we use to give a new formalisation of the infinite adele ring of a number field.
With the help of recent work in local fields, \cite{dFFN24}, we also provide important and necessary topological results, such as the proof that the local ring of integers is compact. 
We note that our work depends on the \href{https://github.com/mariainesdff/LocalClassFieldTheory/blob/9b1150821ebe8a48ac8f5d5ac6b88e0028a8e427/LocalClassFieldTheory/DiscreteValuationRing/ResidueField.lean#L94-L103}{result}, from \cite{dFFN24}, that the $v$-adic ring of integers $\mathcal{O}_v$ for number fields has a finite residue field.
At the time of writing, this result is being upstreamed to \mathlib\footnote{\href{https://github.com/leanprover-community/mathlib4/pull/26537}{https://github.com/leanprover-community/mathlib4/pull/26537}}\textsuperscript{,}\footnote{\href{https://github.com/leanprover-community/mathlib4/pull/26538}{https://github.com/leanprover-community/mathlib4/pull/26538}}.
To avoid duplicating work, we assume this result in our project and leave the proof as a \lstinline{sorry}.

We provide some mathematical preliminaries in Section~\ref{sec:math}, involving the definitions of absolute values, discrete valuations, number fields, and locally compact topological spaces, as well as completions of uniform spaces. 
In Section~\ref{sec:abs_induced}, we discuss families of induced absolute values on a field and how we handle multiple topological instances coming from them in Lean.
Section~\ref{sec:adele_sec} concerns the formalisation of the adele ring of a number field; in particular, the new formalisations of the completion of a number field at an infinite place and the infinite adele ring are detailed in Section~\ref{sec:places}.
We then describe the informal and formal proof that the adele ring is locally compact, along with other topological results, in Section~\ref{sec:locallycompact}. 
Section~\ref{sec:discussion} concludes with a discussion around implementation details.

%% file: sections/preliminaries.tex
\section{Preliminaries}
\label{sec:math}

\subsection{Absolute values and discrete valuations}
\label{sec:valuations}

While absolute values and discrete valuations are often unified into a single treatment, we separate them here due to their distinction within \mathlib.
See \cite[Chapter 2]{CF67} and \cite[Chapter 1]{serre79} for details on the theory of absolute values and discrete valuations respectively, and \cite[Section 2]{dFFN24} for details of integer-valued multiplicative valuations in \mathlib.

\begin{definition}[\href{https://github.com/leanprover-community/mathlib4/blob/caac5b13fb72ba0c5d0b35a0067de108db65e964/Mathlib/Algebra/Order/AbsoluteValue.lean\#L31-L37}{\small\texttt{AbsoluteValue}}]\label{def:abs} 
Let $K$ be a field. We say that the function $|\cdot| : K \to \mathbb{R}_{\ge 0}$ is an \emph{absolute value} on $K$ if, for all $x, y \in K$, we have:
\begin{enumerate}
	\item $|x| = 0$ if and only if $x = 0$;
	\item $|xy| = |x||y|$;
	\item $|x + y| \le |x| + |y|$ (the \emph{triangle inequality}).
\end{enumerate}
If the \emph{ultrametric inequality},
\begin{enumerate}[start=4]
	\item $|x + y| \le \max(|x|, |y|)$,
\end{enumerate}
holds over the triangle inequality, the absolute value is \emph{non-Archimedean}; else, it is \emph{Archimedean}.
\end{definition}

Discrete \emph{additive} valuations on $K$ are a subclass of group homomorphisms $K^{\times}\to(\mathbb{Z}, +)$ that are extended to $K$ as follows.

\begin{definition}[\href{https://github.com/leanprover-community/mathlib4/blob/caac5b13fb72ba0c5d0b35a0067de108db65e964/Mathlib/RingTheory/Valuation/Basic.lean\#L580-L581}{\small\texttt{AddValuation}}]\label{def:discreteaddval} 
The function $a : K \to \mathbb{Z}\cup \{\infty\}$ is a \emph{discrete additive valuation} on $K$ if, for all $x, y \in K$, we have:
\begin{enumerate}
	\item $a(x) = \infty$ if and only if $x = 0$;
	\item $a(xy) = a(x) + a(y)$;
	\item $a(x + y) \ge \min(a(x), a(y))$.
\end{enumerate}
\end{definition}

\begin{example} If $K = \mathbb{Q}$, then the additive $p$-adic valuation $a_p$ (defined on an integer as the largest power of $p$ dividing that integer; extended to $\mathbb{Q}$ by $a_p(a/b) := a_p(a) - a_p(b)$) is an example of a discrete additive valuation. 
\end{example}
We can construct a non-Archimedean absolute value from a discrete additive valuation $a$ by defining $|x|_a = c^{-a(x)}$ if $x \ne 0$ and $|0|_a = 0$, for some chosen real number $c > 1$. 
The $p$-adic absolute value $|x|_p = p^{-a_p(x)}$ is an example of this construction when $K = \mathbb{Q}$. 

Discrete valuations in the literature are typically defined as additive valuations.  
In \mathlib, discrete \emph{multiplicative} valuations have a more developed API than their additive counterparts.
These are abstractions of the absolute values arising from discrete additive valuations that do not require a choice of the base $c$. 
In \mathlib, any abstract additive structure can be turned into a corresponding abstract multiplicative structure, which is in bijection with the original additive type via the inverse morphisms \href{https://github.com/leanprover-community/mathlib4/blob/caac5b13fb72ba0c5d0b35a0067de108db65e964/Mathlib/Algebra/Group/TypeTags.lean\#L81-L82}{\small\texttt{ofAdd}} and \href{https://github.com/leanprover-community/mathlib4/blob/caac5b13fb72ba0c5d0b35a0067de108db65e964/Mathlib/Algebra/Group/TypeTags.lean\#L85}{\small\texttt{toAdd}}. 
For example, the additive group of integers $(\mathbb{Z}, +)$ bijects with the type $\mathbb{Z}_m$, which can be thought of as the multiplicative group $\{x^n \mid n \in \mathbb{Z}\}$, where $x$ is an abstract symbol, as:
\begin{align*}
	\text{\lstinline{ofAdd}} &: (\mathbb{Z}, +) \to(\mathbb{Z}_m, \times); \\
	\text{\lstinline{toAdd}} &: (\mathbb{Z}_m, \times)\to (\mathbb{Z}, +).
\end{align*}
In particular, these maps preserve respective units, generators and preorders, so $\text{\lstinline{ofAdd}}(0) = 1$, $\text{\lstinline{ofAdd}}(1)$ generates $(\mathbb{Z}_m, \times)$, and $\text{\lstinline{ofAdd}}(x) \le \text{\lstinline{ofAdd}}(y)$ if and only if $x \le y$. 
The map \lstinline{ofAdd} can be thought of as $n\mapsto x^n$, where $x$ is a symbol.
To transfer infinite additive values to the multiplicative setting we add a zero term to obtain the type $\mathbb{Z}_{m0}$, which extends the multiplicative structure and the preorder of $\mathbb{Z}_m$. 

\begin{definition}[\href{https://github.com/leanprover-community/mathlib4/blob/caac5b13fb72ba0c5d0b35a0067de108db65e964/Mathlib/RingTheory/Valuation/Basic.lean\#L77-L79}{\small\texttt{Valuation}}]\label{def:discretemulval} 
The function $v : K \to \mathbb{Z}_{m0}$ is a \emph{discrete multiplicative valuation} on $K$ if, for all $x, y \in K$, we have:
\begin{enumerate}
	\item $v(0) = 0$;
	\item $v(1) = 1$;
	\item $v(xy) = v(x)v(y)$;
	\item $v(x + y) \le \max(v(x), v(y))$.
\end{enumerate}
\end{definition}

We emphasise that a discrete multiplicative valuation is not the same as the absolute value $|\cdot|_a$ constructed from a discrete additive valuation $a$. 
The former only takes values in $\mathbb{Z}_{m0}$, while the latter does not; the latter depends on a choice of base $c$, while the former does not. 
However, each additive valuation $a$ also gives rise to a multiplicative valuation $v$ by defining $v(x) = \text{\lstinline{ofAdd}}(-a(x))$ for $x \ne 0$ and $v(0) = 0$.

\begin{example} 
If $K = \mathbb{Q}$, then the $p$-adic valuation $v_p$ is defined by $v_p(x) = \text{\lstinline{ofAdd}}(-a_p(x))$ if $x \ne 0$. 
We see that $x$ is a $p$-adic unit if and only if $v_p(x) = \text{\lstinline{ofAdd}}(0) = 1$, and $v_p(x)\le 1$ if and only if $a_p(x) \ge 0$.  
\end{example}

\begin{definition}\label{def:absequiv} 
Two absolute values (resp. discrete multiplicative valuations) $|\cdot|_1$ and $|\cdot|_2$ (resp. $v_1$ and $v_2$) on $K$ are \emph{equivalent} if $|\cdot|_1 = |\cdot|_2^{\alpha}$ (resp. $v_1 = v_2^{\alpha}$) for some $\alpha \in \mathbb{R}_{>0}$. 
\end{definition}
Equivalent valuations define the same uniform structures on $K$ and lead to the same completions of $K$, so we typically care more about equivalence classes of valuations.

\begin{definition}\label{def:place} 
An \emph{infinite place} of a field $K$ is an equivalence class of Archimedean absolute values on $K$ \emph{(\href{https://github.com/leanprover-community/mathlib4/blob/caac5b13fb72ba0c5d0b35a0067de108db65e964/Mathlib/NumberTheory/NumberField/Embeddings.lean\#L246}{\small\texttt{NumberField.InfinitePlace}})}. 
A \emph{finite place} of $K$ is an equivalence class of discrete multiplicative valuations on $K$.
\end{definition}

\begin{remark} 
Note that both the predicate that two absolute values are equivalent and the type of finite places of a number field have since been included in \mathlib\footnote{\href{https://github.com/leanprover-community/mathlib4/pull/20362}{https://github.com/leanprover-community/mathlib4/pull/20362}}\textsuperscript{,}\footnote{\href{https://github.com/leanprover-community/mathlib4/pull/19667}{https://github.com/leanprover-community/mathlib4/pull/19667}}.
Instead, we worked directly with valuations coming from distinct non-zero prime ideals.
\end{remark}

\subsection{Number fields}

A number field $K$ is a finite-degree field extension of $\mathbb{Q}$, and its ring of integers $\mathcal{O}_K$ is the integral closure of $\mathbb{Z}$ inside $K$. 
If $\mathcal{O}_K$ is not a unique factorisation domain, then the fundamental theorem of arithmetic does not hold at the level of elements, however it does hold at the level of ideals. 
That is, any non-zero ideal of $\mathcal{O}_K$ can be written as a finite product of non-zero prime ideals, unique up to reordering of the prime ideals.

\begin{definition}[$\mathfrak{p}$-adic valuation, \href{https://github.com/leanprover-community/mathlib4/blob/caac5b13fb72ba0c5d0b35a0067de108db65e964/Mathlib/RingTheory/DedekindDomain/AdicValuation.lean\#L296-L298}{\small\texttt{IsDedekindDomain.HeightOneSpectrum.valuation}}]\label{def:Kpadicval} 
The \emph{additive $\mathfrak{p}$-adic  valuation} $a_{\mathfrak{p}}(x)$ on a non-zero prime ideal $\mathfrak{p}$ of $\mathcal{O}_K$ is the largest power of $\mathfrak{p}$ appearing in the factorisation of $x\mathcal{O}_K$ into prime ideals. 
This extends to a discrete additive valuation on $K$ by defining $a_{\mathfrak{p}}(a/b) = a_{\mathfrak{p}}(a) - a_{\mathfrak{p}}(b)$ and $a_{\mathfrak{p}}(0) = \infty$. 
The \emph{$\mathfrak{p}$-adic valuation} $v_{\mathfrak{p}}$ is the multiplicative valuation defined as $v_{\mathfrak{p}}(x) = \text{\emph{\lstinline{ofAdd}}}(-a_{\mathfrak{p}}(x))$ if $x \ne 0$ and $v_{\mathfrak{p}}(0) = 0$.
\end{definition}

\subsection{Uniform spaces and their completions}
\label{sec:uniform}

While topological spaces abstract the notion of continuity, \emph{uniform spaces}, first introduced by \cite{weil1938}, provide the appropriate abstraction of \emph{uniform} continuity.
We describe some brief relevant theory here; for more information, we refer the reader to \cite[Sections~2,~5]{BCM20}.
A uniform space $(X, \mathcal{U})$ consists of a set $X$ and a \emph{uniform structure} $\mathcal{U}$, which is a filter on $X \times X$ satisfying certain properties. 
Uniform continuity of a map $f : (X, \mathcal{U}_X) \to (Y, \mathcal{U}_Y)$ between uniform spaces is understood via the condition that every set in $\mathcal{U}_X$ is also in $(f \times f)^*\mathcal{U}_Y$, where $(f \times f)^*\mathcal{U}_Y$ is the \emph{pullback} filter generated by preimages of elements in $\mathcal{U}_Y$ under $f\times f$.

Uniform spaces come with a powerful completion operation, which makes use of the unified notion of limits that filters represent.
Much like metric spaces can be completed by adding limits of Cauchy sequences, there is an operation sending a uniform space $(X, \mathcal{U})$ to a complete separated uniform space $(\widehat{X}, \widehat{\mathcal{U}})$ by adding limits of \emph{Cauchy filters}, \cite{BCM20}. 
This completion operator satisfies the universal property of \cite[Theorem II.3.7.3, Definition II.3.7.4]{bourbaki66} (see also \cite[Theorem 5.1]{BCM20}), hence it defines a functor from the category of all uniform spaces to the category of complete separated uniform spaces. 
The universal property of the uniform space completion functor means that it represents an equivalence class of \emph{abstract completion operators} as follows.

\begin{definition}[\cite{BCM20}, \href{https://github.com/leanprover-community/mathlib4/blob/caac5b13fb72ba0c5d0b35a0067de108db65e964/Mathlib/Topology/UniformSpace/AbstractCompletion.lean\#L57-L71}{\small\texttt{AbstractCompletion}}]\label{def:abstractcompl} We say that $(Y, \iota : X\to Y)$ is an \emph{abstract completion} of the uniform space $X$ if $Y$ is a complete separated uniform space, $\iota$ has dense image and $\iota$ is uniformly continuous.
\end{definition}

\begin{theorem}[\cite{BCM20}, \href{https://github.com/leanprover-community/mathlib4/blob/caac5b13fb72ba0c5d0b35a0067de108db65e964/Mathlib/Topology/UniformSpace/AbstractCompletion.lean\#L239-L245}{\small\texttt{AbstractCompletion.compareEquiv}}]\label{thm:abstractcompl_iso} Any two abstract completions $(Y_i, \iota_i: X\to Y_i)$, for $i = 1, 2$, of a uniform space $X$ are isomorphic as uniform spaces. 
\end{theorem}

These properties make the uniform space completion operation ideal for formalising the completions of a range of objects, including the number fields that we consider within this paper. 
This theory is contained in \mathlib within the \href{https://github.com/leanprover-community/mathlib4/blob/caac5b13fb72ba0c5d0b35a0067de108db65e964/Mathlib/Topology/UniformSpace/Completion.lean\#L293}{\small\texttt{UniformSpace.Completion}} API, much of which is inherited from the API of the \lstinline{AbstractCompletion} structure. 

A field $K$ can have various sources of uniform structures.
An absolute value $|\cdot|$ on $K$ determines a uniform structure through the filter generated by the sets $\{(x, y)\mid |x - y| <\varepsilon\}$, where $\varepsilon > 0$.
Discrete multiplicative valuations $v$ similarly define uniform spaces $(K, \mathcal{U}_v)$.
These uniform structures determine the completion $K_v$ of a field at each place $v$ via the uniform space completion functor.

\begin{example}
\begin{enumerate}
	\item The uniform structure $\mathcal{U}_{\infty}$ coming from the usual absolute value $|\cdot|_{\infty}$ on $\mathbb{Q}$ defines a uniform space $(\mathbb{Q}, \mathcal{U}_{\infty})$, whose completion is the field $\mathbb{R}$ of real numbers.
	\item The uniform structure $\mathcal{U}_{v_p}$ coming from the multiplicative $p$-adic valuation defines a uniform space $(\mathbb{Q}, \mathcal{U}_{v_p})$, whose completion is the field $\mathbb{Q}_p$ of $p$-adic numbers. The subring $\mathbb{Z}$ of $\mathbb{Q}$ completes to the $p$-adic ring of integers $\mathbb{Z}_p := \{x \in \mathbb{Q}_p \mid v_p(x) \le 1\}$.
\end{enumerate}
\end{example}

\subsection{Locally compact spaces} 

There are many characterisations of a topological space being locally compact. 
The following corresponds to the implementation found in \mathlib.

\begin{definition}[\href{https://github.com/leanprover-community/mathlib4/blob/caac5b13fb72ba0c5d0b35a0067de108db65e964/Mathlib/Topology/Defs/Filter.lean\#L286-L289}{\small\texttt{LocallyCompactSpace}}]\label{def:lc} A topological space $X$ is \emph{locally compact} if, for each neighbourhood $N$ of each $x\in X$, there exists a compact neighbourhood $S \subseteq N$ of $x$.
\end{definition}
It is well known that closed subspaces of locally compact spaces, finite products of locally compact spaces, and infinite products of compact spaces are all locally compact. 

%% file: sections/absuniformspace.tex
\section{Induced absolute values and their pullback uniform structures}
\label{sec:abs_induced}

Throughout Section~\ref{sec:abs_induced} the following variables are in scope.
\begin{lean}
variable {K L : Type*} [Field K] [NormedField L] (v : AbsoluteValue K ℝ)
\end{lean}

\subsection{Handling multiple instances on a type using dependent type synonyms}
\label{sec:withabs}

As described in Section~\ref{sec:math}, there may be multiple distinct \lstinline{UniformSpace} instances on a number field coming from its infinite places. 
However, Lean's type class inference system cannot automatically resolve multiple instances of the same class on a single type. 
If we assign multiple \lstinline{UniformSpace} instances directly to the type \lstinline{K}, we will then be forced to manually specify the desired instance through the use of \lstinline{@} or local \lstinline{let} declarations.

An alternative approach to handling this issue is through the use of type synonyms. 
Type synonyms simply rename a type, however they also allow us to insert dependencies. 
This provides a mechanism through which the type class inference system is able to resolve multiple dependent instances on a type.
For a general semiring \lstinline{R} and ordered semiring \lstinline{S}, we define the type synonym \lsthref{https://github.com/smmercuri/adele-ring_locally-compact/blob/d47637a02a84cf03d6488a4780c4d0399b04278a/AdeleRingLocallyCompact/NumberTheory/NumberField/Completion.lean\#L79}{\small\texttt{WithAbs}} of \lstinline{R}, which depends on an \lstinline{S}-valued absolute value on \lstinline{R}.

\begin{lean}
def (*\lsthref{https://github.com/smmercuri/adele-ring_locally-compact/blob/d47637a02a84cf03d6488a4780c4d0399b04278a/AdeleRingLocallyCompact/NumberTheory/NumberField/Completion.lean\#L79}{WithAbs}*) {R S} [Semiring R] [OrderedSemiring S] : AbsoluteValue R S → Type _ := 
  fun _ => R
\end{lean}
In particular, given a \emph{real} absolute value \lstinline{v} on a field \lstinline{K}, we assign a \lstinline{NormedField} instance to \lstinline{K} by assigning it to its type synonym \lstinline{WithAbs v}.
\begin{lean}
instance (*\lsthref{https://github.com/smmercuri/adele-ring_locally-compact/blob/d47637a02a84cf03d6488a4780c4d0399b04278a/AdeleRingLocallyCompact/NumberTheory/NumberField/Completion.lean\#L97-L100}{normedField}*) : NormedField (WithAbs v) where ...
\end{lean}
Now there is only a \emph{single} \lstinline{NormedField} instance on each \lstinline{WithAbs v}, enabling the type class inference system to automatically resolve the relevant \lstinline{NormedField} and its parent \lstinline{UniformSpace} instance on \lstinline{K} coming from the absolute value \lstinline{v}.
The completion of a field with respect to an absolute value can then be given by applying \lstinline{UniformSpace.Completion}.

\begin{lean}
abbrev (*\lsthref{https://github.com/smmercuri/adele-ring_locally-compact/blob/d47637a02a84cf03d6488a4780c4d0399b04278a/AdeleRingLocallyCompact/NumberTheory/NumberField/Completion.lean\#L139}{AbsoluteValue.Completion}*) := UniformSpace.Completion (WithAbs v)
\end{lean}

See Section~\ref{sec:multiple_instances} for a comparison of the dependent type synonym approach to other approaches for handling the multiple instance problem. 

\subsection{Induced absolute values} 

A family of uniform structures on a field $K$ can be generated through field embeddings $K \hookrightarrow L$ of $K$ into a field $L$, where $L$ has an absolute value $|\cdot|_L : L \to \mathbb{R}$ as follows.

\begin{definition}\label{def:induced_abs} 
Let $K$ and $L$ be fields and let $|\cdot|_L : L \to\mathbb{R}$ be an absolute value on $L$. 
If $\sigma : K \hookrightarrow L$ is a field embedding, then the $\sigma$\emph{-induced} absolute value $|\cdot|_{\sigma} : K \to \mathbb{R}$ is defined by $|x|_{\sigma} = |\sigma(x)|_L$ for all $x \in K$.
\end{definition}
The uniform structure $\mathcal{U}_{\sigma}$ on $K$ given by a $(\sigma : K \hookrightarrow L)$-induced absolute value $|\cdot|_{\sigma}$ is precisely the pullback, $(\sigma \times \sigma)^*\mathcal{U}_L$, of the uniform structure $\mathcal{U}_L$ given by the absolute value $|\cdot|_L$. 
In \mathlib, this property is encoded in the \href{https://github.com/leanprover-community/mathlib4/blob/caac5b13fb72ba0c5d0b35a0067de108db65e964/Mathlib/Topology/UniformSpace/UniformEmbedding.lean#L32-L35}{\small\texttt{UniformInducing}} definition.

\begin{lean}
theorem (*\lsthref{https://github.com/smmercuri/adele-ring_locally-compact/blob/d47637a02a84cf03d6488a4780c4d0399b04278a/AdeleRingLocallyCompact/NumberTheory/NumberField/Completion.lean\#L127-L128}{WithAbs.uniformInducing\_of\_comp}*) {f : WithAbs v →+* L}
    (h : ∀ x, ‖f x‖ = v x) : UniformInducing f := ...
\end{lean}

\subsection{Pullbacks preserve the completable topological field property}
\label{sec:completabletopfield}
In general, the completion $\widehat{K}$ of a separated uniform space $(K, \mathcal{U})$ may not be a field, even if $K$ is a field. 
The necessary and sufficient condition required to guarantee this is that $(K, \mathcal{U})$ is a \emph{completable topological field}, given in \cite[Theorem III.6.8.7]{bourbaki66} and also \cite[Theorem 5.4]{BCM20}, which requires the inverse map $x \mapsto x^{-1}$ on $K$ to preserve the Cauchy property of filters which do not have a cluster point at zero.
The type class \href{https://github.com/leanprover-community/mathlib4/blob/caac5b13fb72ba0c5d0b35a0067de108db65e964/Mathlib/Topology/Algebra/UniformField.lean\#L52-L53}{\small\texttt{CompletableTopField}} encodes this property within \mathlib.

Let $(L, \mathcal{U}_L)$ be a completable topological field.
The pullback $(\sigma \times\sigma)^*\mathcal{U}_L$ uniform structure on $K$ under a field embedding $\sigma : K \hookrightarrow L$ defines a completable topological field.
\begin{lean}
theorem (*\lsthref{https://github.com/smmercuri/adele-ring_locally-compact/blob/d47637a02a84cf03d6488a4780c4d0399b04278a/AdeleRingLocallyCompact/Topology/UniformSpace/Basic.lean\#L26-L34}{UniformInducing.completableTopField}*) {K L : Type*} [Field L] 
    [UniformSpace L] [CompletableTopField L] [Field K] [UniformSpace K] 
    [T0Space K] {f : K →+* L}  (hf : UniformInducing f) :
    CompletableTopField K := ...
\end{lean}
Since the uniform structure of a ($\sigma : K \hookrightarrow L)$-induced absolute value of a field $K$ coincides with the pullback $(\sigma \times\sigma)^*\mathcal{U}_L$, we see that $(K, \mathcal{U}_{\sigma})$ is a completable topological field.
The completion of $K$ with respect to a $\sigma$-induced absolute value is therefore a field.
 

%% file: sections/adelering.tex
\section{The adele ring of a number field}
\label{sec:adele_sec}

Ostrowski's theorem \cite{ostrowski16} states that the only infinite place of $\mathbb{Q}$ is represented by the usual absolute value $|\cdot|_{\infty}$ and all the finite places are represented by the $p$-adic valuations $v_p$.
This result has been formalised and is part of a later version of \mathlib\footnote{\href{https://github.com/leanprover-community/mathlib4/pull/17138}{https://github.com/leanprover-community/mathlib4/pull/17138}}.
The \emph{infinite adele ring} $\mathbb{A}_{\mathbb{Q}, \infty}$ of $\mathbb{Q}$ is the completion $\mathbb{R}$ of $\mathbb{Q}$ at the single infinite place, and the \emph{finite adele ring} $\mathbb{A}_{\mathbb{Q}, f}$ is the restricted product of the completions of $\mathbb{Q}$ over primes $p$,
\begin{align}\label{eq:Qfiniteadelering}
	\mathbb{A}_{\mathbb{Q}, f} := \left\{(x_p)_p \in \prod_p \mathbb{Q}_p \mathrel{\bigg|} x_p \in \mathbb{Z}_p\ \text{for all but finitely many $p$}\right\}.
\end{align}
The \emph{adele ring} of $\mathbb{Q}$ is the product of the infinite and finite adele rings, $\mathbb{A}_{\mathbb{Q}} := \mathbb{A}_{\mathbb{Q}, \infty} \times \mathbb{A}_{\mathbb{Q}, f}$. 
The story is similar for general number fields, with the infinite places coming from the complex absolute value $|\cdot|_{\mathbb{C}}$ and the finite places being represented by $\mathfrak{p}$-adic valuations of non-zero prime ideals $\mathfrak{p}$ of $\mathcal{O}_K$. 

\subsection{The infinite adele ring of a number field}
\label{sec:places}

Throughout Section~\ref{sec:places} the following variables are in scope.
\begin{lean}
variable (K : Type*) [Field K] [NumberField K] (v : NumberField.InfinitePlace K)
\end{lean}

\subsubsection{Infinite places of a number field}
\label{sec:infiniteplaces}

There is only a single infinite place of $\mathbb{Q}$ because there is only a single field embedding of $\mathbb{Q}$ into $\mathbb{C}$.
In general, number fields $K$ have a finite number of distinct embeddings. 
For example, the embeddings of $K = \mathbb{Q}(\alpha)$ correspond to the Galois action permuting the roots of the minimal polynomial of $\alpha$ over $\mathbb{Q}$. 
Distinct embeddings $\sigma : K \hookrightarrow \mathbb{C}$ generate non-equivalent Archimedean absolute values on $K$ via the $\sigma$-induced absolute values $|x|_{\sigma} := |\sigma(x)|_{\mathbb{C}}$. 
These define all the infinite places of a number field, the set of which is denoted by $\Sigma_{K, \infty}$.

The type \href{https://github.com/leanprover-community/mathlib4/blob/caac5b13fb72ba0c5d0b35a0067de108db65e964/Mathlib/NumberTheory/NumberField/Embeddings.lean\#L246}{\small\texttt{NumberField.InfinitePlace}} in \mathlib\footnote{\href{https://github.com/leanprover-community/mathlib3/pull/17844}{https://github.com/leanprover-community/mathlib3/pull/17844}} encodes the infinite places of a number field $K$.
Each term \lstinline{v} consists of an absolute value 
\lstinline{v.val : AbsoluteValue K ℝ}, and a proof (\lstinline{‖v.embedding x‖ = v.val x} for all \lstinline{x : K}) that \lstinline{v.val} is induced by some \lstinline{v.embedding : K →+* ℂ}.

\subsubsection{Completing a number field at an infinite place}
\label{sec:infinitecompl}

Infinite places $v \in \Sigma_{K, \infty}$ define distinct uniform structures $\mathcal{U}_v$ on $K$ through their associated absolute values and we can therefore complete $K$ at $v$ to obtain $K_v$, as described in Section~\ref{sec:uniform}. 
Formally, we do this using \lstinline{AbsoluteValue.Completion} of Section~\ref{sec:withabs}.
\begin{lean}
abbrev (*\lsthref{https://github.com/smmercuri/adele-ring_locally-compact/blob/d47637a02a84cf03d6488a4780c4d0399b04278a/AdeleRingLocallyCompact/NumberTheory/NumberField/Completion.lean\#L197}{NumberField.InfinitePlace.Completion}*) := v.val.Completion
\end{lean}
Because these absolute values are $(\sigma : K \hookrightarrow\mathbb{C})$-induced, we have that $K$ is a completable topological field with respect to each $\mathcal{U}_v$ (Section~\ref{sec:completabletopfield}), and so the completion $K_v$ is a field.
Moreover, the absolute value on $K$ extends to $K_v$ to yield a \lstinline{NormedField} instance.

\begin{lean}
instance : NormedField v.Completion :=
  letI := (WithAbs.uniformInducing_of_comp v.norm_embedding_eq).completableTopField
  UniformSpace.Completion.instNormedFieldOfCompletableTopField (WithAbs v.val)
\end{lean}

\subsubsection{The infinite adele ring}
\label{sec:infiniteadele}

\begin{definition}[\href{https://github.com/smmercuri/adele-ring_locally-compact/blob/d47637a02a84cf03d6488a4780c4d0399b04278a/AdeleRingLocallyCompact/NumberTheory/NumberField/InfiniteAdeleRing.lean\#L54}{\small\texttt{NumberField.InfiniteAdeleRing}}] The infinite adele ring $\mathbb{A}_{K, \infty}$ of a number field $K$ is the product topological ring over all infinite place completions of $K$:
\begin{equation}\label{eq:AKinf}
	\mathbb{A}_{K, \infty} := \prod_{v\in\Sigma_{K, \infty}} K_v.
\end{equation}
\end{definition}

\begin{lean}
def (*\lsthref{https://github.com/smmercuri/adele-ring_locally-compact/blob/d47637a02a84cf03d6488a4780c4d0399b04278a/AdeleRingLocallyCompact/NumberTheory/NumberField/InfiniteAdeleRing.lean\#L54}{NumberField.InfiniteAdeleRing}*) := (v : InfinitePlace K) → v.Completion

instance : TopologicalRing (InfiniteAdeleRing K) := Pi.instTopologicalRing
\end{lean}

\subsection{The finite adele ring of a number field}
\label{sec:finiteadele}

The finite adele ring was formalised in \mathlib, following \cite{defrutosfernandez22}, for any Dedekind domain $R$ of Krull dimension one and its field of fractions $K$, by taking the restricted product over the non-zero prime spectrum $\Spec(R)$ of $R$. 
Formally, we work in the same generality.
However, for simplicity our informal discussion restricts to $R = \mathcal{O}_K$, as required by our main results on number fields.
Throughout Section~\ref{sec:finiteadele} the following variables are in scope.
\begin{lean}
variable (R : Type*) [CommRing R] [IsDomain R] [IsDedekindDomain R] 
  (K : Type*) [Field K] [Algebra R K] [IsFractionRing R K]
\end{lean}

\subsubsection{Finite places of a number field}

Discrete valuations of global fields and their completions at finite places were previously formalised in \cite{defrutosfernandez22, dFFN24}. 
The set of all finite places of a number field $K$, denoted by $\Sigma_{K, f}$, is indexed by the non-zero prime ideals $\mathfrak{p}\in\Spec(\mathcal{O}_K)$ through their $\mathfrak{p}$-adic valuations (Definition~\ref{def:Kpadicval}).
The \mathlib structure \lstinline{Valuation K ℤₘ₀} represents discrete multiplicative valuations of $K$. 
The typing \lstinline[mathescape]{v : IsDedekindDomain.HeightOneSpectrum ($\mathcal{O}$ K)} represents the relation $\mathfrak{p}\in\Spec(\mathcal{O}_K)$ and its multiplicative valuation is \lstinline{v.valuation : Valuation K ℤₘ₀}. 
Moreover, each prime ideal term \lstinline{v} defines a uniform structure on $K$, via \href{https://github.com/leanprover-community/mathlib4/blob/caac5b13fb72ba0c5d0b35a0067de108db65e964/Mathlib/RingTheory/DedekindDomain/AdicValuation.lean\#L356-L357}{\small\texttt{v.adicValued.toUniformSpace}}, which is then used to formalise the $v$-adic completion $K_v$ of $K$ as \href{https://github.com/leanprover-community/mathlib4/blob/caac5b13fb72ba0c5d0b35a0067de108db65e964/Mathlib/RingTheory/DedekindDomain/AdicValuation.lean\#L365-L366}{\small\texttt{v.adicCompletion~K}}. 
The $v$-adic ring of integers, $\mathcal{O}_v := \{x \in K_v \mid v(x) \le \texttt{ofAdd}(0)\}$, 
is the subring \href{https://github.com/leanprover-community/mathlib4/blob/caac5b13fb72ba0c5d0b35a0067de108db65e964/Mathlib/RingTheory/DedekindDomain/AdicValuation.lean\#L393-L394}{\small\texttt{v.adicCompletionIntegers~K}}, 
which is a discrete valuation ring with unique maximal ideal $\mathfrak{m}_v := \{x \in \mathcal{O}_v \mid v(x) < \texttt{ofAdd}(0)\}$.
This is a principal ideal generated by a choice of uniformizer $\pi$, which is any element such that $v(\pi) = \texttt{ofAdd}(-1)$ (this corresponds to having additive valuation $a(\pi) = 1$).
The formal maximal ideal of $\mathcal{O}_v$ is \href{https://github.com/leanprover-community/mathlib4/blob/caac5b13fb72ba0c5d0b35a0067de108db65e964/Mathlib/Topology/Algebra/Valued/ValuedField.lean\#L370}{\small\texttt{Valued.maximalIdeal (v.adicCompletion K)}}.

\subsubsection{The finite adele ring}

\begin{definition}[\href{https://github.com/leanprover-community/mathlib4/blob/caac5b13fb72ba0c5d0b35a0067de108db65e964/Mathlib/RingTheory/DedekindDomain/FiniteAdeleRing.lean\#L305}{\small\texttt{DedekindDomain.FiniteAdeleRing}}]
The \emph{finite adele ring} $\mathbb{A}_{K, f}$ of a number field $K$,
\begin{align}\label{eq:Kfiniteadelering}
	\mathbb{A}_{K, f} := \left\{(x_v)_v \in \prod_{v \in \Sigma_{K, f}} K_v \mathrel{\Bigg|} x_v \in \mathcal{O}_v\ \text{for all but finitely many $v$}\right\},
\end{align}
is the topological ring whose topology is generated by the basis $\left\{\prod_{v}q\mathcal{O}_v \mathrel{\big|} 0 \ne q \in \mathcal{O}_K\right\}$ of neighbourhoods of zero.
\end{definition}
In the version of \mathlib used in this project, the finite adele ring of any Dedekind domain \lstinline{R} and its field of fractions \lstinline{K} was introduced by the work of \cite{defrutosfernandez22}. 
Its topological space structure was incorporated into \mathlib in a later PR\footnote{\href{https://github.com/leanprover-community/mathlib4/pull/14176}{https://github.com/leanprover-community/mathlib4/pull/14176}}; this makes use of the \href{https://github.com/leanprover-community/mathlib4/blob/caac5b13fb72ba0c5d0b35a0067de108db65e964/Mathlib/Topology/Algebra/Nonarchimedean/Bases.lean\#L197-L205}{\small\texttt{SubmodulesRingBasis}} formalism which constructs a \lstinline{TopologicalSpace} instance from a basis of neighbourhoods of zero, which is compatible with the underlying ring structure.  
In other words, it also generates the appropriate \lstinline{TopologicalRing} instance. 

\begin{remark}\label{rmk:nhds}
Open sets containing zero in the finite adele ring are generated by sets of the form $\prod_{v \in S} B_{r_v}(0) \times \prod_{v \notin S} \mathcal{O}_v$, where $S$ is finite and $0 < r_v \in \mathbb{Z}$. 
The family of all such sets is parameterised by $0\ne q \in K$ via $\prod_v q\mathcal{O}_v$.
However, the neighbourhoods of zero filter is generated through a basis by upwards inclusion.
Since $\mathcal{O}_v \subseteq q\mathcal{O}_v$ for any $v$ dividing the denominator of $q$, it suffices to consider only $q \in \mathcal{O}_K$ in order to give a basis of the neighbourhoods of zero.
\end{remark}

\begin{remark}\label{rmk:refactor_1}
Since the completion of this project, both the definition and the topology of the finite adele ring in \mathlib have seen a significant refactor to use the \lstinline{RestrictedProduct} API\footnote{\href{https://github.com/leanprover-community/mathlib4/pull/20021}{https://github.com/leanprover-community/mathlib4/pull/20021}}\textsuperscript{,}\footnote{\href{https://github.com/leanprover-community/mathlib4/pull/23542}{https://github.com/leanprover-community/mathlib4/pull/23542}}.
The impact of this refactor on the local compactness of the finite adele ring is discussed in Remark~\ref{rmk:refactor_2}.
\end{remark}

\subsection{The adele ring} 
\label{sec:adele}

\begin{definition}[\href{https://github.com/smmercuri/adele-ring_locally-compact/blob/d47637a02a84cf03d6488a4780c4d0399b04278a/AdeleRingLocallyCompact/NumberTheory/NumberField/AdeleRing.lean\#L42}{\small\texttt{NumberField.AdeleRing}}]\label{def:adelering} The adele ring $\mathbb{A}_K$ of a number field is the product topological ring of the infinite and the finite adele rings:
\begin{equation}\label{eq:adelering}
	\mathbb{A}_K := \mathbb{A}_{K, \infty}\times\mathbb{A}_{K, f}.
\end{equation}
\end{definition}
\begin{lean}[mathescape]
def (*\lsthref{https://github.com/smmercuri/adele-ring_locally-compact/blob/d47637a02a84cf03d6488a4780c4d0399b04278a/AdeleRingLocallyCompact/NumberTheory/NumberField/AdeleRing.lean\#L42}{NumberField.AdeleRing}*) (K : Type*) [Field K] [NumberField K] := 
  InfiniteAdeleRing K × DedekindDomain.FiniteAdeleRing ($\mathcal{O}$ K) K
\end{lean}

%% file: sections/locallycompact.tex
\section{Local compactness of the adele ring of a number field}
\label{sec:locallycompact}

This section contains an informal and formal proof that the adele ring of a number field is locally compact, using \cite[Chapter 2]{CF67} as source. 
Full details of the formal proof can be found in the source code.
The informal proofs are written so as to align with the formal proofs.

\subsection{Local compactness of the infinite adele ring}
\label{sec:locallycompactinfinite}

Throughout Section~\ref{sec:locallycompactinfinite} the following variables are in scope.
\begin{lean}
variable {K L : Type*} [Field K] [NormedField L] [CompleteSpace L] 
  {v : AbsoluteValue K ℝ} {f : WithAbs v →+* L} 
\end{lean}
We first show that the completion of a number field at an infinite place is locally compact.
As in Sections~\ref{sec:abs_induced} and \ref{sec:adele_sec}, our strategy is to first formalise the local compactness in higher abstraction for fields with an associated induced absolute value, from which the result for number fields and infinite places can be derived. 

\begin{theorem}\label{thm:induced_compl_lc} Let $K$ and $L$ be fields, let $|\cdot|_L : L \to \mathbb{R}$ be an absolute value on $L$, and assume that $L$ is locally compact and complete with respect to the topology induced by $|\cdot|_L$. If $\sigma : K \hookrightarrow L$ is a field embedding, then the completion of $K$ with respect to the $\sigma$-induced absolute value $|\cdot|_{\sigma}$ is locally compact.
\end{theorem}

\begin{proof} Since the uniform structure given by $|\cdot|_{\sigma}$ is the pullback uniform structure, $\sigma$ is uniformly continuous. Therefore $\sigma$ extends to an embedding $\sigma_v : K_v \hookrightarrow L$ by the universal property of uniform completions. Further, $\sigma$ and $\sigma_v$ are both isometries; thus, $K_v$ has a closed image under $\sigma_v$ and so is locally compact.
\end{proof}

\begin{lean}[escapeinside={(*}{*)}]
namespace AbsoluteValue.Completion

abbrev (*\lsthref{https://github.com/smmercuri/adele-ring_locally-compact/blob/d47637a02a84cf03d6488a4780c4d0399b04278a/AdeleRingLocallyCompact/NumberTheory/NumberField/Completion.lean\#L150-L152}{extensionEmbeddingOfComp}*) (h : ∀ x, ‖f x‖ = v x) : v.Completion →+* L :=
  UniformSpace.Completion.extensionHom _
    (WithAbs.uniformInducing_of_comp h).uniformContinuous.continuous
\end{lean}

\begin{lean}[escapeinside={(*}{*)}]
theorem (*\lsthref{https://github.com/smmercuri/adele-ring_locally-compact/blob/d47637a02a84cf03d6488a4780c4d0399b04278a/AdeleRingLocallyCompact/NumberTheory/NumberField/Completion.lean\#L161-L168}{extensionEmbeddingOfComp\_dist\_eq}*) (h : ∀ x, ‖f x‖ = v x) 
    (x y : v.Completion) :
    dist (extensionEmbeddingOfComp h x) (extensionEmbeddingOfComp h y) =
      dist x y := by ...
\end{lean}

\begin{lean}[escapeinside={(*}{*)}]
theorem (*\lsthref{https://github.com/smmercuri/adele-ring_locally-compact/blob/d47637a02a84cf03d6488a4780c4d0399b04278a/AdeleRingLocallyCompact/NumberTheory/NumberField/Completion.lean\#L172-L174}{isometry\_extensionEmbeddingOfComp}*) (h : ∀ x, ‖f x‖ = v x) :
    Isometry (extensionEmbeddingOfComp h) :=
  Isometry.of_dist_eq <| extensionEmbeddingOfComp_dist_eq h
\end{lean}

\begin{lean}[escapeinside={(*}{*)}]
theorem (*\lsthref{https://github.com/smmercuri/adele-ring_locally-compact/blob/d47637a02a84cf03d6488a4780c4d0399b04278a/AdeleRingLocallyCompact/NumberTheory/NumberField/Completion.lean\#L178-L180}{closedEmbedding\_extensionEmbeddingOfComp}*) (h : ∀ x, ‖f x‖ = v x) :
    ClosedEmbedding (extensionEmbeddingOfComp h) :=
  (isometry_extensionEmbeddingOfComp h).closedEmbedding
\end{lean}

\begin{lean}[escapeinside={(*}{*)}]
theorem (*\lsthref{https://github.com/smmercuri/adele-ring_locally-compact/blob/d47637a02a84cf03d6488a4780c4d0399b04278a/AdeleRingLocallyCompact/NumberTheory/NumberField/Completion.lean\#L184-L186}{locallyCompactSpace}*) [LocallyCompactSpace L] (h : ∀ x, ‖f x‖ = v x) :
    LocallyCompactSpace (v.Completion) :=
  (closedEmbedding_extensionEmbeddingOfComp h).locallyCompactSpace

end AbsoluteValue.Completion
\end{lean}

\begin{corollary} \label{thm:infcompl_lc} The completion $K_v$ of a number field $K$ at $v\in\Sigma_{K, \infty}$ is locally compact.
\end{corollary}

\begin{lean}
instance (*\lsthref{https://github.com/smmercuri/adele-ring_locally-compact/blob/d47637a02a84cf03d6488a4780c4d0399b04278a/AdeleRingLocallyCompact/NumberTheory/NumberField/Completion.lean\#L209-L210}{NumberField.InfinitePlace.Completion.locallyCompactSpace}*) 
    (v : InfinitePlace K) : LocallyCompactSpace (v.Completion) :=
  AbsoluteValue.Completion.locallyCompactSpace v.norm_embedding_eq
\end{lean}

\begin{theorem}\label{thm:infiniteadele_lc} The infinite adele ring $\mathbb{A}_{K, \infty}$ of a number field $K$ is locally compact.
\end{theorem}

\begin{proof} Since each completion $K_v$ is locally compact by Corollary~\ref{thm:infcompl_lc} then, as a finite product of locally compact spaces, so is $\mathbb{A}_{K, \infty}$.
\end{proof}

\begin{lean}
theorem (*\lsthref{https://github.com/smmercuri/adele-ring_locally-compact/blob/d47637a02a84cf03d6488a4780c4d0399b04278a/AdeleRingLocallyCompact/NumberTheory/NumberField/InfiniteAdeleRing.lean\#L79-L80}{NumberField.InfiniteAdeleRing.locallyCompactSpace}*) [NumberField K] : 
    LocallyCompactSpace (InfiniteAdeleRing K) :=
  Pi.locallyCompactSpace_of_finite
\end{lean}
 
\subsection{Local compactness of the finite adele ring}
\label{sec:locallycompactfinite}

As a subring of an infinite product, it is not immediately clear that the finite adele ring is locally compact.
However, each $\mathcal{O}_v$ is \emph{compact}, and each finite adele is in $\mathcal{O}_v$ at infinitely-many places.
This is the reason why the finite adele ring is locally compact.
To prove this, we first show that $\mathcal{O}_v$ is compact and that $K_v$ is locally compact in Section~\ref{sec:compactO}.
Then in Sections~\ref{sec:finiteSadele} and \ref{sec:finiteSadele_cover} we show that the finite adele ring admits an open cover of locally compact \emph{finite $S$-adele rings}.

\begin{remark}\label{rmk:refactor_2}
Since the completion of this project, the finite adele ring has seen a significant refactor to use the \lstinline{RestrictedProduct} API\footnote{\href{https://github.com/leanprover-community/mathlib4/pull/20021}{https://github.com/leanprover-community/mathlib4/pull/20021}}\textsuperscript{,}\footnote{\href{https://github.com/leanprover-community/mathlib4/pull/23542}{https://github.com/leanprover-community/mathlib4/pull/23542}}.
While the compactness results of Section~\ref{sec:compactO} remain essential, the results from Sections~\ref{sec:finiteSadele} and \ref{sec:finiteSadele_cover} on the local compactness of the finite adele ring have now been generalised to a corresponding result on restricted products.
Some of the auxiliary objects used in this section, such as \href{https://github.com/leanprover-community/mathlib4/blob/caac5b13fb72ba0c5d0b35a0067de108db65e964/Mathlib/RingTheory/DedekindDomain/FiniteAdeleRing.lean\#L68-L69}{\small\texttt{ProdAdicCompletions}} and \href{https://github.com/leanprover-community/mathlib4/blob/caac5b13fb72ba0c5d0b35a0067de108db65e964/Mathlib/RingTheory/DedekindDomain/FiniteAdeleRing.lean\#L44-L45}{\small\texttt{FiniteIntegralAdeles}}, have since been removed from \mathlib.
Nevertheless, we retain the formal $S$-adele argument here as a record of the first formalised proof of this fact.
\end{remark} 

Throughout Section~\ref{sec:locallycompactfinite} the following variables are in scope.
\begin{lean}
variable {R : Type*} [CommRing R] [IsDedekindDomain R] 
  (K : Type*) [Field K] [Algebra R K] [IsFractionRing R K] [NumberField K]
  (v : HeightOneSpectrum R) (S : Finset (HeightOneSpectrum R))
\end{lean}

\subsubsection{Compactness of the $v$-adic ring of integers}
\label{sec:compactO}

\begin{theorem}\label{thm:Ov_compact} The ring of integers $\mathcal{O}_v$ of the completion $K_v$ of a number field $K$ at $v \in \Sigma_{K, f}$ is compact.
\end{theorem}

\begin{proof} We show equivalently that $\mathcal{O}_v$ is complete and totally bounded.  
As a closed subset of the complete space $K_v$, $\mathcal{O}_v$ is complete. 
It is totally bounded if, for any $\gamma\in \mathbb{Z}_m$, there exists a finite cover of $\mathcal{O}_v$ of open balls $B_{\gamma}(t) := \{x\mid v(x - t) < \gamma\}$ of radius $\gamma$. 
It suffices to check this for $\gamma \le 1$, in which case we take the finitely-many representatives $t_i$ of $\mathcal{O}_v/\mathfrak{m}_v^{-\text{\lstinline{toAdd}}(\gamma) + 1}$. 
The balls $B_{\gamma}(t_i)$ then cover $\mathcal{O}_v$.
\end{proof}

\begin{lean}
namespace IsDedekindDomain.HeightOneSpectrum.adicCompletionIntegers

theorem (*\lsthref{https://github.com/smmercuri/adele-ring_locally-compact/blob/d47637a02a84cf03d6488a4780c4d0399b04278a/AdeleRingLocallyCompact/RingTheory/DedekindDomain/AdicValuation.lean\#L249-L256}{isClosed}*) : IsClosed (v.adicCompletionIntegers K).carrier := ... 
\end{lean}

\begin{lean}
theorem (*\lsthref{https://github.com/smmercuri/adele-ring_locally-compact/blob/d47637a02a84cf03d6488a4780c4d0399b04278a/AdeleRingLocallyCompact/RingTheory/DedekindDomain/AdicValuation.lean\#L284-L286}{totallyBounded}*) : TotallyBounded (v.adicCompletionIntegers K).carrier := ...
\end{lean}

\begin{lean}
theorem (*\lsthref{https://github.com/smmercuri/adele-ring_locally-compact/blob/d47637a02a84cf03d6488a4780c4d0399b04278a/AdeleRingLocallyCompact/RingTheory/DedekindDomain/AdicValuation.lean\#L292-L293}{isCompact}*) : IsCompact (v.adicCompletionIntegers K).carrier :=
  isCompact_iff_totallyBounded_isComplete.2 
    ⟨totallyBounded K v, (isClosed K v).isComplete⟩
\end{lean}

\begin{lean}
instance (*\lsthref{https://github.com/smmercuri/adele-ring_locally-compact/blob/d47637a02a84cf03d6488a4780c4d0399b04278a/AdeleRingLocallyCompact/RingTheory/DedekindDomain/AdicValuation.lean\#L295-L296}{compactSpace}*) : CompactSpace (v.adicCompletionIntegers K) :=
  isCompact_iff_compactSpace.1 <| isCompact K v
\end{lean}

The proof of Theorem~\ref{thm:Ov_compact} requires that the residue field $\mathcal{O}_v/\mathfrak{m}_v$ is \emph{finite}.
As described in the introduction, this has been \href{https://github.com/mariainesdff/LocalClassFieldTheory/blob/9b1150821ebe8a48ac8f5d5ac6b88e0028a8e427/LocalClassFieldTheory/DiscreteValuationRing/ResidueField.lean#L94-L103}{formalised elsewhere}, \cite{dFFN24}, and is currently in the process of being upstreamed to \mathlib\footnote{\href{https://github.com/leanprover-community/mathlib4/pull/26537}{https://github.com/leanprover-community/mathlib4/pull/26537}}\textsuperscript{,}\footnote{\href{https://github.com/leanprover-community/mathlib4/pull/26538}{https://github.com/leanprover-community/mathlib4/pull/26538}}. 
We assume the statement with a \lstinline{sorry} proof in our work to avoid duplication.
However, we prove under its assumption that $\mathcal{O}_v/\mathfrak{m}_v^n$ is finite for any integer $n \ge 0$.
To do this we define the finite-coefficient map \href{https://github.com/smmercuri/adele-ring_locally-compact/blob/d47637a02a84cf03d6488a4780c4d0399b04278a/AdeleRingLocallyCompact/RingTheory/DedekindDomain/AdicValuation.lean\#L224-L227}{\small\texttt{toFiniteCoeffs n h}$\pi$} sending $x \in \mathcal{O}_v/\mathfrak{m}_v^n$ to an $n$-tuple $(x_1, ..., x_n)$ of $v$-adic digits in its $\pi$-adic expansion, where $\pi$ is a uniformizer.
This gives an injective function $\mathcal{O}_v/\mathfrak{m}_v^n \to (\mathcal{O}_v/\mathfrak{m}_v)^n$.

\begin{lean} 
theorem (*\lsthref{https://github.com/smmercuri/adele-ring_locally-compact/blob/d47637a02a84cf03d6488a4780c4d0399b04278a/AdeleRingLocallyCompact/RingTheory/DedekindDomain/AdicValuation.lean\#L229-L237}{toFiniteCoeffs\_injective}*) {π : v.adicCompletionIntegers K} (n : ℕ) 
    (hπ : IsUniformizer π.val) : (toFiniteCoeffs n hπ).Injective := ...
\end{lean}

\begin{lean}
instance (*\lsthref{https://github.com/smmercuri/adele-ring_locally-compact/blob/d47637a02a84cf03d6488a4780c4d0399b04278a/AdeleRingLocallyCompact/RingTheory/DedekindDomain/AdicValuation.lean\#L240-L243}{quotient\_maximalIdeal\_pow\_finite}*) {π : v.adicCompletionIntegers K} (n : ℕ)
    (hπ : IsUniformizer π.val) :
    Finite (v.adicCompletionIntegers K / 
    	(Valued.maximalIdeal (v.adicCompletion K)) ^ n) :=
  Finite.of_injective _ (toFiniteCoeffs_injective n hπ)
\end{lean}
% EDITOR'S NOTE(Rob): I've cut
% end adicCompletionIntegers
% from the end of this listings block

\begin{theorem}\label{thm:fincompl_lc} The completion $K_v$ of a number field at $v\in\Sigma_{K, f}$ is locally compact.
\end{theorem}

\begin{proof}
It is enough to show that $0$ has a compact neighbourhood because we can translate and dilate this to be contained in neighbourhoods around other points.
The neighbourhood $B_1(0)$ of $0$ is a closed subspace of the compact space $\mathcal{O}_v$, so it is compact.
\end{proof}

\begin{lean}[escapeinside={(*}{*)}]
theorem (*\lsthref{https://github.com/smmercuri/adele-ring_locally-compact/blob/d47637a02a84cf03d6488a4780c4d0399b04278a/AdeleRingLocallyCompact/RingTheory/DedekindDomain/AdicValuation.lean\#L307-L310}{adicCompletion.isCompact\_nhds\_zero}*) {γ : ℤₘ₀(*$^{\times}$*)} (hγ : γ ≤ 1) : 
    IsCompact { y : v.adicCompletion K | Valued.v y < γ } :=
  (isCompact K v).of_isClosed_subset (isClosed_nhds_zero K v γ)
    <| fun _ hx => le_of_lt (lt_of_lt_of_le (Set.mem_setOf.1 hx) hγ)
\end{lean}

\begin{lean}[escapeinside={(*}{*)}]
instance (*\lsthref{https://github.com/smmercuri/adele-ring_locally-compact/blob/d47637a02a84cf03d6488a4780c4d0399b04278a/AdeleRingLocallyCompact/RingTheory/DedekindDomain/AdicValuation.lean\#L315-L317}{adicCompletion.locallyCompactSpace}*) : 
    LocallyCompactSpace (v.adicCompletion K) :=
  (isCompact_nhds_zero K v le_rfl).locallyCompactSpace_of_mem_nhds_of_addGroup
    <| (hasBasis_nhds_zero K v).mem_of_mem le_rfl
\end{lean}
\subsubsection{The finite $S$-adele ring}
\label{sec:finiteSadele}

Associated to any element $x$ of the finite adele ring is its \emph{support} -- a finite set of places $S_x\subseteq\Sigma_{K, f}$ for which $x_v\in \mathcal{O}_v$ if and only if $v\notin S_x$. It is the dependence of the set $S_x$ on $x$ that makes the local compactness of the finite adele ring difficult to view. 
On the other hand, if we fix a finite set $S\subseteq\Sigma_{K, f}$, and consider only the finite adeles $x$ for which $S_x \subseteq S$, then this is an easier space to understand as locally compact. These are the finite $S$-adeles.

\begin{definition}[\href{https://github.com/smmercuri/adele-ring_locally-compact/blob/d47637a02a84cf03d6488a4780c4d0399b04278a/AdeleRingLocallyCompact/RingTheory/DedekindDomain/FinsetAdeleRing.lean\#L186}{\small\texttt{DedekindDomain.FinsetAdeleRing}}]\label{def:finiteSadelering} Let $S\subseteq\Sigma_{K, f}$ be a finite set of places of a number field $K$. The finite $S$-adele ring $\mathbb{A}_{K, S, f}$ is the topological ring defined as
\begin{align*}
	\mathbb{A}_{K, S, f} := \left\{(x_v)_v \in \prod_{v\in\Sigma_{K, f}} K_v \mathrel{\bigg|} x_v \in \mathcal{O}_v\ \text{for all $v\notin S$}\right\}.
\end{align*}
\end{definition}
Mathematically, $\mathbb{A}_{K, S, f}$ is given the subspace topology of $\prod_v K_v$, but it can also be viewed as a subspace of $\mathbb{A}_{K, f}$, so it has two immediate sources of topologies. 
These define the \emph{same} topology on $\mathbb{A}_{K, S, f}$, which is a crucial step in the overall proof that $\mathbb{A}_{K, f}$ is locally compact.

We formalise $\mathbb{A}_{K, S, f}$ as a subtype of $\prod_v K_v$, distinct from $\mathbb{A}_{K, f}$.
It is given the subspace topology of $\prod_v K_v$.
It does not inherit the subspace topology of $\mathbb{A}_{K, f}$, but we prove in Section~\ref{sec:finiteSadele_cover} that the $\mathbb{A}_{K, f}$-induced topology matches the inherited $\prod_v K_v$-subspace topology.

% \newpage
\begin{lean}
namespace DedekindDomain
\end{lean}

\begin{lean}
def (*\lsthref{https://github.com/smmercuri/adele-ring_locally-compact/blob/d47637a02a84cf03d6488a4780c4d0399b04278a/AdeleRingLocallyCompact/RingTheory/DedekindDomain/FinsetAdeleRing.lean\#L148-L149}{IsFinsetAdele}*) (x : ProdAdicCompletions R K) :=
  ∀ v ∉ S, x v ∈ v.adicCompletionIntegers K
\end{lean}

\begin{lean}
def (*\lsthref{https://github.com/smmercuri/adele-ring_locally-compact/blob/d47637a02a84cf03d6488a4780c4d0399b04278a/AdeleRingLocallyCompact/RingTheory/DedekindDomain/FinsetAdeleRing.lean\#L186}{FinsetAdeleRing}*) := {x : ProdAdicCompletions R K // IsFinsetAdele S x}
\end{lean}
We construct a \lstinline{Subring (ProdAdicCompletions R K)} term  \href{https://github.com/smmercuri/adele-ring_locally-compact/blob/d47637a02a84cf03d6488a4780c4d0399b04278a/AdeleRingLocallyCompact/RingTheory/DedekindDomain/FinsetAdeleRing.lean\#L197-L203}{\small\texttt{FinsetAdeleRing.subring R K}}, with carrier \lstinline[mathescape]{$\{$x | IsFinsetAdele S x$\}$}.
This is used to infer the subspace \lstinline{TopologicalRing} instance on \lstinline{FinsetAdeleRing R K}.
Defining \lstinline{FinsetAdeleRing} as a type and then constructing a \lstinline{Subring} term separately is modelled after PRs\footnote{\href{https://github.com/leanprover-community/mathlib4/pull/12386}{https://github.com/leanprover-community/mathlib4/pull/12386}}\textsuperscript{,}\footnote{\href{https://github.com/leanprover-community/mathlib4/pull/13021}{https://github.com/leanprover-community/mathlib4/pull/13021}} which converted the ring of integers of a number field and the finite adele ring to types from \lstinline{Subalgebra} and \lstinline{Subring} terms respectively. 
The former PR demonstrated significant performance improvements in \mathlib, with files seeing up to 71.9\% faster build time.

\begin{theorem}\label{thm:finiteSadele_lc} The finite $S$-adele ring $\mathbb{A}_{K, S, f}$ of a number field $K$ is locally compact.
\end{theorem}

\begin{proof} The finite $S$-adele ring is the product
\[
	\prod_{v \in S} K_v \times \prod_{v\notin S} \mathcal{O}_v,
\]
whose first factor is locally compact as a finite product of locally compact spaces, and whose second factor is locally compact as an infinite product of compact spaces.
\end{proof}

The above proof hides some details that are trivial mathematically, but require care in formalisation. 
This comes down to two concerns: (1) how we formalise the space $\prod_{v \in S} K_v \times~\prod_{v\notin S} \mathcal{O}_v$; (2) the sense in which we identify $\mathbb{A}_{K, S, f}$ and $\prod_{v \in S} K_v \times\prod_{v\notin S} \mathcal{O}_v$.

For (1), $\prod_{v \in S} K_v \times \prod_{v\notin S} \mathcal{O}_v$ can be formalised as a type, given a product topological space instance, which we denote mathematically also by $\prod_{v \in S} K_v \times \prod_{v\notin S} \mathcal{O}_v$.
Or, it can be a subtype of $\prod_{v\in S} K_v\times \prod_{v\notin S} K_v$, given a \lstinline{Subtype} topological space instance, which we denote mathematically by $\widehat{\mathcal{O}}_S$. 
We shall use both, because the former is really the space whose local compactness is deduced in the proof of Theorem~\ref{thm:finiteSadele_lc} above, but the latter simplifies the transfer of this local compactness to $\mathbb{A}_{K, S, f}$ considerably. 
The spaces $\prod_{v\in S} K_v \times \prod_{v\notin S} K_v$, $\prod_{v \in S} K_v \times \prod_{v\notin S} \mathcal{O}_v$ and $\widehat{\mathcal{O}}_S$ are formalised respectively as follows.

\begin{lean}
def (*\lsthref{https://github.com/smmercuri/adele-ring_locally-compact/blob/d47637a02a84cf03d6488a4780c4d0399b04278a/AdeleRingLocallyCompact/RingTheory/DedekindDomain/FinsetAdeleRing.lean\#L72-L73}{ProdAdicCompletions.FinsetProd}*) := ((v : S) → v.val.adicCompletion K) × 
  ((v : {v // v ∉ S}) → v.val.adicCompletion K)
\end{lean}
  
\begin{lean}
def (*\lsthref{https://github.com/smmercuri/adele-ring_locally-compact/blob/d47637a02a84cf03d6488a4780c4d0399b04278a/AdeleRingLocallyCompact/RingTheory/DedekindDomain/FinsetAdeleRing.lean\#L89-L90}{FinsetIntegralAdeles}*) := ((v : S) → v.val.adicCompletion K) × 
  ((v : {v // v ∉ S}) → v.val.adicCompletionIntegers K)
\end{lean}

\begin{lean}
def (*\lsthref{https://github.com/smmercuri/adele-ring_locally-compact/blob/d47637a02a84cf03d6488a4780c4d0399b04278a/AdeleRingLocallyCompact/RingTheory/DedekindDomain/FinsetAdeleRing.lean\#L98}{FinsetIntegralAdeles.Subtype}*) := 
  {x : FinsetProd R K S // ∀ v, x.2 v ∈ v.val.adicCompletionIntegers K}
\end{lean}

The identification concern of (2) can then be achieved through composing two homeomorphisms $t_1$ and $t_2$, as illustrated in the commutative diagram of Figure~\ref{fig:Sadelearc}.
Here, $p : \prod_v K_v \to \text{\lstinline{Prop}}; x \mapsto \forall v\notin S, x_v \in \mathcal{O}_v$ is the subtype condition defining $\mathbb{A}_{K, S, f}$ and $q : (x_1, x_2) \mapsto \forall v\notin S, (x_2)_v \in \mathcal{O}_v$ is the subtype condition defining $\widehat{\mathcal{O}}_S$. 
The \emph{partial function} $\langle\cdot, p\rangle: \prod_v K_v\rightharpoonup\mathbb{A}_{K, S, f}$ is the function defined on each $x\in\prod_v K_v$ satisfying $p(x)$, sending it to the corresponding element in the subtype $\mathbb{A}_{K, S, f}$, and $\langle \cdot, q\rangle$ is defined analogously. 
The homeomorphism $t_{\pi}$ is given in \mathlib as \href{https://github.com/leanprover-community/mathlib4/blob/caac5b13fb72ba0c5d0b35a0067de108db65e964/Mathlib/Topology/Homeomorph.lean\#L665-L673}{\small\texttt{Homeomorph.piEquivPiSubtypeProd}}.
\begin{figure}[h]
\centering
\begin{tikzpicture}[>=Stealth]
    	\node (AKS) at (0,0) {\( \mathbb{A}_{K, S, f}\)};
    	\node (K_hat) at (0,3) {\(\displaystyle\prod_v K_v\)};
	\node (K_hat_prod) at (4, 3) {\(\displaystyle\prod_{v\in S} K_v \times \prod_{v\notin S}K_v\)};
    	\node (K_hat_S_st) at (4,0) {\(\widehat{\mathcal{O}}_S\)};
	\node (K_hat_S) at (8, 3) {\(\displaystyle\prod_{v\in S}K_v \times\prod_{v\notin S}\mathcal{O}_v\)};

    	\draw[left to-] (AKS) -- (K_hat) node[midway, left] {$\langle\cdot, p\rangle$};
	\draw[<->] (K_hat) -- (K_hat_prod) node[midway, above] {$t_{\pi}$} node[midway, below] {$\cong$};
	\draw[-left to] (K_hat_prod) -- (K_hat_S_st) node[midway, right] {$\langle\cdot, q\rangle$};
	\draw[dotted, <->] (AKS) edge node[midway, below] {$t_1$}  (K_hat_S_st);
	\draw[dotted, <->] (K_hat_S_st) -- (K_hat_S) node[midway, below right]{$t_2$};
\end{tikzpicture}
\caption{\label{fig:Sadelearc} Type relationships in the formal proof for the local compactness of $\mathbb{A}_{K, S, f}$.}
\end{figure}

Given the formalisations of the local compactness of $K_v$ and the compactness of $\mathcal{O}_v$ in Section~\ref{sec:compactO}, the local compactness of $\prod_v \in K_v \times \prod_{v\notin S} \mathcal{O}_v$ can be easily inferred.
\begin{lean}
instance (*\lsthref{https://github.com/smmercuri/adele-ring_locally-compact/blob/d47637a02a84cf03d6488a4780c4d0399b04278a/AdeleRingLocallyCompact/RingTheory/DedekindDomain/FinsetAdeleRing.lean\#L137-L138}{FinsetIntegralAdeles.locallyCompactSpace}*) : 
    LocallyCompactSpace (FinsetIntegralAdeles R K S) :=
  Prod.locallyCompactSpace _ _
\end{lean}
The homeomorphism $t_2$ is defined by \lstinline{Prod.mk} and \lstinline{Subtype.mk} maps, whose continuity is straightforward to show. 
This then gives the local compactness of $\widehat{\mathcal{O}}_S$.
\begin{lean}[escapeinside={(*}{*)}]
def (*\lsthref{https://github.com/smmercuri/adele-ring_locally-compact/blob/d47637a02a84cf03d6488a4780c4d0399b04278a/AdeleRingLocallyCompact/RingTheory/DedekindDomain/FinsetAdeleRing.lean\#L127-L131}{FinsetIntegralAdeles.subtypeHomeomorph}*) :
    Subtype R K S ≃(*$_t$*) FinsetIntegralAdeles R K S where ...

instance (*\lsthref{https://github.com/smmercuri/adele-ring_locally-compact/blob/d47637a02a84cf03d6488a4780c4d0399b04278a/AdeleRingLocallyCompact/RingTheory/DedekindDomain/FinsetAdeleRing.lean\#L141-L142}{FinsetIntegralAdeles.locallyCompactSpaceSubtype}*) : 
    LocallyCompactSpace (Subtype R K S) :=
  (subtypeHomeomorph R K S).closedEmbedding.locallyCompactSpace
\end{lean}
It is easy to see that $p = q\circ t_{\pi}$ on defined domains, so the homeomorphism $t_1$ can be lifted from $t_{\pi}$ using \href{https://github.com/leanprover-community/mathlib4/blob/caac5b13fb72ba0c5d0b35a0067de108db65e964/Mathlib/Topology/Homeomorph.lean\#L458-L463}{\small\texttt{Homeomorph.subtype}} and used to infer that $\mathbb{A}_{K, S, f}$ is locally compact. 

\begin{lean}[escapeinside={(*}{*)}]
namespace FinsetAdeleRing
\end{lean}

\begin{lean}[escapeinside={(*}{*)}]
def (*\lsthref{https://github.com/smmercuri/adele-ring_locally-compact/blob/d47637a02a84cf03d6488a4780c4d0399b04278a/AdeleRingLocallyCompact/RingTheory/DedekindDomain/FinsetAdeleRing.lean\#L218-L221}{homeomorphSubtype}*) : FinsetAdeleRing R K S ≃(*$_t$*) FinsetIntegralAdeles.Subtype R K S :=
  (Homeomorph.piEquivPiSubtypeProd _ _).subtype <| fun _ =>
    ⟨fun hx v => hx v.1 v.2, fun hx v hv => hx ⟨v, hv⟩⟩
\end{lean}

\begin{lean}[escapeinside={(*}{*)}]
instance (*\lsthref{https://github.com/smmercuri/adele-ring_locally-compact/blob/d47637a02a84cf03d6488a4780c4d0399b04278a/AdeleRingLocallyCompact/RingTheory/DedekindDomain/FinsetAdeleRing.lean\#L224-L225}{locallyCompactSpace}*) : LocallyCompactSpace (FinsetAdeleRing R K S) :=
  (homeomorphSubtype R K S).closedEmbedding.locallyCompactSpace
\end{lean}

\subsubsection{Using the finite $S$-adele rings to cover the finite adele ring}
\label{sec:finiteSadele_cover}

\begin{theorem}\label{thm:finiteadele_lc} The finite adele ring $\mathbb{A}_{K, f}$ of a number field $K$ is locally compact.
\end{theorem}

\begin{proof}
Let $x \in \mathbb{A}_{K, f}$ and $N$ be a neighbourhood of $x$. 
We have $x \in\mathbb{A}_{K, S_x, f}$, where $S_x$ is the set of finite places for which $x \notin \mathcal{O}_v$. 
By the definition of the topology on the finite adele ring, $\mathbb{A}_{K, S_x, f}$ is a neighbourhood of $x$.
Moreover, it is locally compact by Theorem~\ref{thm:finiteSadele_lc}. 
The neighbourhood $N\cap\mathbb{A}_{K, S_x, f}$ of $x$ inside the finite $S$-adele ring therefore contains a compact neighbourhood $M$ of $x$. 
The inclusion map $\mathbb{A}_{K, S_x, f}\hookrightarrow \mathbb{A}_{K, f}$ sends neighbourhoods to neighbourhoods and is continuous open, so the image of $M$ into $\mathbb{A}_{K, f}$ is a compact neighbourhood of $x$ in $\mathbb{A}_{K, f}$.
\end{proof}

\begin{remark}
The proof in this section shows local compactness by direct satisfaction of Definition~\ref{def:lc}. 
In our code, we also provide a \href{https://github.com/smmercuri/adele-ring_locally-compact/blob/d47637a02a84cf03d6488a4780c4d0399b04278a/AdeleRingLocallyCompact/RingTheory/DedekindDomain/FiniteAdeleRingAlt.lean}{simplified proof} which takes advantage of the fact that, in a topological ring, we need only show that $0$ has a compact neighbourhood, which is given by $\prod_v \mathcal{O}_v$ in this case.
\end{remark}

The subtlety here is that, in Theorem~\ref{thm:finiteSadele_lc}, we showed that the finite $S$-adele ring with the subspace topology of $\prod_v K_v$ is locally compact, whereas in the above proof, we are viewing it as a subspace of $\mathbb{A}_{K, f}$. 
The key point is that these topologies coincide.
The $\mathbb{A}_{K, f}$-subspace topology is defined by inducing through the embedding $\iota(S) : \mathbb{A}_{K, S, f} \hookrightarrow \mathbb{A}_{K, f}$.
\begin{lean}
instance : Algebra (FinsetAdeleRing R K S) (FiniteAdeleRing R K) where ...

local notation "ι(" S ")" => algebraMap (FinsetAdeleRing R K S) (FiniteAdeleRing R K)
  
theorem (*\lsthref{https://github.com/smmercuri/adele-ring_locally-compact/blob/d47637a02a84cf03d6488a4780c4d0399b04278a/AdeleRingLocallyCompact/RingTheory/DedekindDomain/FinsetAdeleRing.lean\#L260-L264}{algebraMap\_injective}*) : Function.Injective ι(S) := ...
\end{lean}
As constructed, \lstinline{FinsetAdeleRing R K S} has the \lstinline{ProdAdicCompletions R K} subspace topology.
In \mathlib, the property that the topology induced through a function coincides with a space's given topology is encoded in the \href{https://github.com/leanprover-community/mathlib4/blob/caac5b13fb72ba0c5d0b35a0067de108db65e964/Mathlib/Topology/Defs/Induced.lean\#L101-L103}{\small\texttt{Inducing}} structure. 
The map $\iota(S)$ is \lstinline{Inducing} if and only if the images of neighbourhoods in $\mathbb{A}_{K, S, f}$ and preimages of neighbourhoods in $\mathbb{A}_{K, f}$ remain neighbourhoods under $\iota(S)$.
\begin{lean}
theorem (*\lsthref{https://github.com/smmercuri/adele-ring_locally-compact/blob/d47637a02a84cf03d6488a4780c4d0399b04278a/AdeleRingLocallyCompact/RingTheory/DedekindDomain/FinsetAdeleRing.lean\#L302-L329}{algebraMap\_image\_mem\_nhds}*) (x : FinsetAdeleRing R K S)
    {U : Set (FinsetAdeleRing R K S)} (h : U ∈ nhds x) :
    ι(S) '' U ∈ nhds (ι(S) x) := by ... 
\end{lean}
    
\begin{lean}
theorem (*\lsthref{https://github.com/smmercuri/adele-ring_locally-compact/blob/d47637a02a84cf03d6488a4780c4d0399b04278a/AdeleRingLocallyCompact/RingTheory/DedekindDomain/FinsetAdeleRing.lean\#L334-L347}{mem\_nhds\_comap\_algebraMap}*) (x : FinsetAdeleRing R K S)
    {U : Set (FinsetAdeleRing R K S)} (h : U ∈ Filter.comap ι(S) (nhds (ι(S) x))) :
    U ∈ nhds x := by ...
\end{lean}
    
\begin{lean}
theorem (*\lsthref{https://github.com/smmercuri/adele-ring_locally-compact/blob/d47637a02a84cf03d6488a4780c4d0399b04278a/AdeleRingLocallyCompact/RingTheory/DedekindDomain/FinsetAdeleRing.lean\#L353-L356}{algebraMap\_inducing}*) : Inducing ι(S) := by
  refine inducing_iff_nhds.2 
    (fun x => Filter.ext (fun U => ⟨fun hU => ⟨ι(S) '' U,  ?_⟩, 
      mem_nhds_comap_algebraMap x⟩))
  exact ⟨algebraMap_image_mem_nhds x hU, 
    by rw [(algebraMap_injective R K S).preimage_image]⟩
\end{lean}
Maps that are \lstinline{Inducing} are also open and continuous, so we can use $\iota(S)$ to push forward compact neighbourhoods from the finite $S$-adele ring to the finite adele ring as in the informal proof of Theorem~\ref{thm:finiteadele_lc}.
The formal proof may then be given as follows.

\newpage
\begin{lean}
instance (*\lsthref{https://github.com/smmercuri/adele-ring_locally-compact/blob/d47637a02a84cf03d6488a4780c4d0399b04278a/AdeleRingLocallyCompact/RingTheory/DedekindDomain/FinsetAdeleRing.lean\#L375-L383}{FiniteAdeleRing.locallyCompactSpace}*) : 
    LocallyCompactSpace (FiniteAdeleRing R K) := by
  refine LocallyCompactSpace.mk <| fun x N hN => let S := support x; ?_
  have h := (algebraMap_inducing R K S).nhds_eq_comap (ofFiniteAdeleSupport x)
  let ⟨M, hM⟩ := (FinsetAdeleRing.locallyCompactSpace R K S).local_compact_nhds
    (ofFiniteAdeleSupport x) _ (h ▸ Filter.preimage_mem_comap hN)
  refine ⟨ι(S) '' M, ?_, Set.image_subset_iff.2 hM.2.1,
    (algebraMap_inducing R K S).isCompact_iff.1 hM.2.2⟩
  have h := algebraMap_range_mem_nhds (ofFiniteAdeleSupport x)
  exact (algebraMap_inducing R K S).map_nhds_of_mem _ h ▸ Filter.image_mem_map hM.1

end DedekindDomain
\end{lean}

\subsection{Local compactness of the adele ring}

\begin{theorem}\label{thm:adelering_lc} The adele ring $\mathbb{A}_K$ of a number field $K$ is locally compact.
\end{theorem}

\begin{proof} As the product of the infinite adele ring and the finite adele ring, which are each locally compact (Theorems~\ref{thm:infiniteadele_lc} and \ref{thm:finiteadele_lc}, respectively), the adele ring is also locally compact.
\end{proof}

\begin{lean}
instance (*\lsthref{https://github.com/smmercuri/adele-ring_locally-compact/blob/d47637a02a84cf03d6488a4780c4d0399b04278a/AdeleRingLocallyCompact/NumberTheory/NumberField/AdeleRing.lean\#L70-L71}{NumberField.AdeleRing.locallyCompactSpace}*) (K : Type*) [Field K] 
    [NumberField K] : LocallyCompactSpace (AdeleRing K) := 
  Prod.locallyCompactSpace _ _
\end{lean}

%% file: sections/discussion.tex
\section{Discussion}
\label{sec:discussion}

\subsection{A comparison of approaches for handling multiple instances}
\label{sec:multiple_instances}
There are many cases in mathematics where we would like multiple distinct instances of the same class on a type.
An example appeared in this research, where we had distinct uniform structures coming from the infinite places of a number field. 
There are a number of methods for assigning such instances in Lean, where the type class inference system expects only a single instance of a particular class to be assigned to a type.
The approach we took to resolve this issue in Section~\ref{sec:withabs} was to define the type synonym \lstinline{WithAbs}.
We will compare the type synonym approach with two other approaches in this section. 

Throughout Section~\ref{sec:multiple_instances} the following variables are in scope.
\begin{lean}
variable {K : Type*} [Field K] (v : AbsoluteValue K ℝ) 
\end{lean}

\subsubsection{Type synonyms}
\label{sec:discussion_synonym}

The basic usage of a type synonym is to create a copy of an already existing type.
This enables one to assign an additional instance of a class to that type.
The archetype of this in \mathlib is \href{https://github.com/leanprover-community/mathlib4/blob/caac5b13fb72ba0c5d0b35a0067de108db65e964/Mathlib/Order/Synonym.lean\#L134-L135}{\small\texttt{Lex}}, which is defined simply as \lstinline{Lex α := α} and is used to give a type its lexicographic order.

It is possible to assign a family of instances by defining dependent type synonyms. 
This was the approach we took for \lstinline{WithAbs}, which redefined a semiring as one that is dependent on absolute values.
This approach creates separate indexed types, so that a single instance from the family is assigned to each, and type class search resolves automatically.
In our application, we required \lstinline{UniformSpace} instances coming from absolute values of a field so that we could apply \lstinline{UniformSpace.Completion} as follows.
\begin{lean}
instance (*\lsthref{https://github.com/smmercuri/adele-ring_locally-compact/blob/d47637a02a84cf03d6488a4780c4d0399b04278a/AdeleRingLocallyCompact/NumberTheory/NumberField/Completion.lean\#L97-L100}{WithAbs.normedField}*) : NormedField (WithAbs v) where ...

abbrev (*\lsthref{https://github.com/smmercuri/adele-ring_locally-compact/blob/d47637a02a84cf03d6488a4780c4d0399b04278a/AdeleRingLocallyCompact/NumberTheory/NumberField/Completion.lean\#L139}{AbsoluteValue.Completion}*) := UniformSpace.Completion (WithAbs v)
\end{lean}
Notice that we do not need to tell the inference system which \lstinline{UniformSpace} instance we really mean.
In addition, type class search resolves automatically in dependent results. 
For example, we can state \lstinline{CompletableTopField K} and apply \lstinline{UniformSpace.Completion.instField} in the following without an explicit \lstinline{UniformSpace} instance.
\begin{lean}
instance [CompletableTopField (WithAbs v)] : Field v.Completion :=
  UniformSpace.Completion.instField
\end{lean}

\subsubsection{Dependent constructors to a type class}

Another approach to assigning multiple instances is to define an alternate constructor to the class we want to assign.
In the case where we have an indexed family of instances, then the alternate constructor should depend on this index.
A previous version of our work used this approach in order to specify instances coming from absolute values. 
In this case we define a constructor to the \lstinline{NormedField} class from an absolute value.
\begin{lean}
instance (*\lsthref{https://github.com/smmercuri/adele-ring_locally-compact/blob/d47637a02a84cf03d6488a4780c4d0399b04278a/AdeleRingLocallyCompact/NumberTheory/NumberField/CompletionAlt.lean\#L73-L82}{AbsoluteValue.normedFieldCons}*) : NormedField K where
  norm := v
  ...
\end{lean}
The required \lstinline{UniformSpace} instance on \lstinline{K} associated to \lstinline{v} can automatically be inferred from \lstinline{v.normedFieldCons}.
However, type class search cannot automatically find \lstinline{v.normedFieldCons}, even when it is made an instance, due to the expectation that \lstinline{K} have only a single \lstinline{NormedField} instance.
This approach has the drawback that we must provide the \lstinline{v.normedFieldCons} instance whenever required, as follows.
\begin{lean}
abbrev (*\lsthref{https://github.com/smmercuri/adele-ring_locally-compact/blob/d47637a02a84cf03d6488a4780c4d0399b04278a/AdeleRingLocallyCompact/NumberTheory/NumberField/CompletionAlt.lean\#L84-L86}{AbsoluteValue.CompletionCons}*) := 
  letI := v.normedFieldCons -- Explicit instance required
  UniformSpace.Completion K

instance [letI := v.normedFieldCons; CompletableTopField K] :
    Field v.CompletionCons :=
  letI := v.normedFieldCons -- Explicit instance required
  UniformSpace.Completion.instField
\end{lean}

\subsubsection{Data-carrying type class}

A third approach to assigning an indexed family of instances is to define a new class altogether, which equips the type with a choice of index.
This is similar to the type synonym approach but the data is stored within a class now, as opposed to a new type.

\begin{lean}
class (*\lsthref{https://github.com/smmercuri/adele-ring_locally-compact/blob/d47637a02a84cf03d6488a4780c4d0399b04278a/AdeleRingLocallyCompact/NumberTheory/NumberField/CompletionAlt.lean\#L121-L122}{WithAbsReal}*) (R : Type*) [Semiring R] where
  v : AbsoluteValue R ℝ
\end{lean}

\begin{lean}
instance [WithAbsReal K] : NormedField K := WithAbsReal.v.normedFieldCons
\end{lean}
Then we make an explicit instance of this type class within the completion operation.
\begin{lean}
abbrev (*\lsthref{https://github.com/smmercuri/adele-ring_locally-compact/blob/d47637a02a84cf03d6488a4780c4d0399b04278a/AdeleRingLocallyCompact/NumberTheory/NumberField/CompletionAlt.lean\#L130-L132}{AbsoluteValue.CompletionClass}*) :=
  letI := WithAbsReal.mk v -- Explicit instance required
  UniformSpace.Completion K
\end{lean}
As in the dependent constructor approach, any instance of a parent class of \lstinline{NormedField} needs to be explicitly provided whenever required.
\begin{lean}
instance [letI := WithAbsReal.mk v; CompletableTopField K] :
    Field v.CompletionClass :=
  letI := WithAbsReal.mk v -- Explicit instance required
  UniformSpace.Completion.instField
\end{lean}

Out of the three approaches described, the type synonym approach optimises type class resolution.
This helps to ensure robust further usage, development and maintenance.

\subsection{Completing a number field at an infinite place through subfields}
\label{sec:subfield_completion}

In Section~\ref{sec:infinitecompl}, we described the formalisation of the completion of a number field at an infinite place.
The approach there hinged on the properties of the uniform space given by $\sigma$-induced absolute values, where $\sigma : K \hookrightarrow \mathbb{C}$.
In this section, we describe an alternate approach where we inject $K$ onto its image first, to obtain a \lstinline{Subfield ℂ} term, and then apply \lstinline{UniformSpace.Completion} to the subfield term.

\begin{lean}
variable {K : Type*} [Field K] (v : InfinitePlace K)
\end{lean}

\begin{lean}
def (*\lsthref{https://github.com/smmercuri/adele-ring_locally-compact/blob/d47637a02a84cf03d6488a4780c4d0399b04278a/AdeleRingLocallyCompact/NumberTheory/NumberField/CompletionAlt.lean\#L175-L177}{NumberField.InfinitePlace.subfield}*) : Subfield ℂ where
  toSubring := v.embedding.range ... 
\end{lean}
  
\begin{lean}
def (*\lsthref{https://github.com/smmercuri/adele-ring_locally-compact/blob/d47637a02a84cf03d6488a4780c4d0399b04278a/AdeleRingLocallyCompact/NumberTheory/NumberField/CompletionAlt.lean\#L180-L185}{toSubfield}*) : K →+* v.subfield where ...
\end{lean}

\begin{lean}
def (*\lsthref{https://github.com/smmercuri/adele-ring_locally-compact/blob/d47637a02a84cf03d6488a4780c4d0399b04278a/AdeleRingLocallyCompact/NumberTheory/NumberField/CompletionAlt.lean\#L195}{NumberField.InfinitePlace.CompletionSubfield}*) := 
  UniformSpace.Completion v.subfield
\end{lean}

\begin{lean}
instance : Field v.CompletionSubfield := inferInstance
\end{lean}

By injecting $K$ to its image, we remove the issues around multiple \lstinline{UniformSpace} instances that we described in Section~\ref{sec:multiple_instances}.
As a result, the type class inference system is immediately able to infer necessary instances such as \lstinline{UniformSpace} and \lstinline{Field}.

The subfield completion defined by \lstinline{v.CompletionSubfield} is isomorphic to \lstinline{v.Completion} of Section~\ref{sec:infinitecompl} as uniform spaces by Theorem~\ref{thm:abstractcompl_iso}, because their constructions both define abstract completions of $K$.
The canonical norm on \lstinline{v.CompletionSubfield} is the \emph{complex} absolute value, but any $x \in K$ first embeds into \lstinline{v.subfield} via $\sigma$ before being coerced to \lstinline{v.CompletionSubfield}, 
so $x \in K$ still has the value $|\sigma(x)|_{\mathbb{C}}$. 
\begin{lean}
instance : Coe K v.CompletionSubfield where
  coe := (UniformSpace.Completion.coe' v.subfield) ∘ v.toSubfield
\end{lean}

However, this approach has a drawback.
Any uniform completion $(Y, \iota : K \to Y)$ of $K$ comes with the property that, for any uniformly continuous homomorphism $f : K \to~Z$ from $K$ to a complete separated uniform field $Z$, there exists a uniformly continuous homomorphism $\hat{f}:~Y \to Z$ such that $\hat{f}(\iota(x)) = f(x)$ for all $x \in K$. 
This map is given by \href{https://github.com/leanprover-community/mathlib4/blob/caac5b13fb72ba0c5d0b35a0067de108db65e964/Mathlib/Topology/Algebra/UniformRing.lean\#L148-L168}{\small\texttt{UniformSpace.Completion.extensionHom}} in \mathlib.
In this alternate construction, however, we can only extend functions from \lstinline{v.subfield} to \lstinline{v.CompletionSubfield}, and not those from $K$ without further work.
Generally, in bypassing the use of \lstinline{UniformSpace.Completion} directly on $K$, we do not have immediate access to parts of the \lstinline{UniformSpace.Completion} API that transfer properties and functions of $K$ to its completion.
This makes it more difficult to maintain a useful API for \lstinline{v.CompletionSubfield}.

\subsection{Formalising the infinite adele ring as a mixed space}
\label{sec:mixedSpace}

An infinite place $v \in \Sigma_{K, \infty}$ is \emph{real} if the image of its associated embedding $K \hookrightarrow \mathbb{C}$ lies entirely within  $\mathbb{R}$, in which case $K_v = \mathbb{R}$, otherwise it is \emph{complex} and $K_v = \mathbb{C}$.
Therefore the infinite adele ring can also be viewed as $\mathbb{R}^{r_1} \times \mathbb{C}^{r_2}$, where $r_1$ and $r_2$ are the number of real and complex places respectively.
This is the space with local notation \href{https://github.com/leanprover-community/mathlib4/blob/caac5b13fb72ba0c5d0b35a0067de108db65e964/Mathlib/NumberTheory/NumberField/CanonicalEmbedding/Basic.lean\#L184-L185}{\small\texttt{E K}} in \mathlib (redefined as \lstinline{NumberField.mixedEmbedding.mixedSpace} in later versions).

We chose not to use \lstinline{mixedSpace} as the formalisation of the infinite adele ring for a number of reasons.
For one, it would mean the adele ring would be defined as a triple product over the two subtypes of \lstinline{InfinitePlace} and the type of prime ideals on $\mathcal{O}_K$, which is more cumbersome.
One could also define the completion of \lstinline{K} at \lstinline{v : InfinitePlace K} as \lstinline{if v.IsReal then ℝ else ℂ}, but this is neither definitionally equal to \lstinline{ℝ} nor \lstinline{ℂ}.
Secondly, as discussed in Section~\ref{sec:subfield_completion}, the \lstinline{UniformSpace.Completion} API comes with useful constructions such as extensions of homomorphisms that allow us to transfer certain functions on $K$ to its completion.
We would not have immediate access to these with the \lstinline{mixedSpace} formalisation.

In our work, we show isometric ring isomorphisms between $K_v$ and $\mathbb{R}$ or $\mathbb{C}$ as appropriate, as well as an isomorphism between the infinite adele ring and the mixed space. 

\begin{lean}
def (*\lsthref{https://github.com/smmercuri/adele-ring_locally-compact/blob/d47637a02a84cf03d6488a4780c4d0399b04278a/AdeleRingLocallyCompact/NumberTheory/NumberField/InfiniteAdeleRing.lean\#L84-L96}{InfiniteAdeleRing.ringEquivMixedSpace}*) (K : Type*) [Field K] :
    InfiniteAdeleRing K ≃+* ({w : InfinitePlace K // IsReal w} → ℝ) × 
      ({w : InfinitePlace K // IsComplex w} → ℂ) := ...
\end{lean}

\subsection{Future work}

This work establishes a \href{https://github.com/ImperialCollegeLondon/FLT/blob/c982409a2c783dde82f6c871102b6500d992b96d/FLT/NumberField/AdeleRing.lean\#L29-L38}{key result} in the early stages of the five-year effort to formalise Fermat's Last Theorem. 
It also represents a step towards the formalisation of Tate's thesis in \mathlib, where the adele ring was used to establish the foundations of the Langlands program for $GL(1)$. 
Other theoretical ingredients for Tate's thesis, such as Haar measures \cite{vandoorn21}, are now also available in \mathlib.
We list some ordered future work to follow directly on from this research.
\begin{itemize}
	\item Show that the adele ring of a function field is locally compact.
	\item Show that $K$ is a discrete cocompact subgroup of $\mathbb{A}_K$.
	\item Prove that the idele group $\mathbb{I}_K$ is locally compact and $K^{\times}$ is a discrete cocompact subgroup of $\mathbb{I}_K$.
	\item Formalise adelic Hecke characters.
	\item Formalise Tate's thesis.
\end{itemize}

%% file: main.bib
@book{bourbaki66,
	author = {Bourbaki, Nicolas},
	mrclass = {54.00 (00.00)},
	mrnumber = {0205210},
	pages = {vii+437},
	publisher = {Hermann Paris; Addison-Wesley Publishing Co. Reading Mass.-London-Don Mills Ont.},
	title = {Elements of mathematics. {G}eneral topology. {P}art 1},
	year = {1966}}

@inproceedings{BCM20,
	address = {New York, NY, USA},
	author = {Buzzard, Kevin and Commelin, Johan and Massot, Patrick},
	booktitle = {Proceedings of the 9th ACM SIGPLAN International Conference on Certified Programs and Proofs},
	doi = {10.1145/3372885.3373830},
	isbn = {9781450370974},
	location = {New Orleans, LA, USA},
	numpages = {14},
	pages = {299--312},
	publisher = {Association for Computing Machinery},
	series = {CPP 2020},
	title = {Formalising perfectoid spaces},
	url = {https://doi.org/10.1145/3372885.3373830},
	year = {2020}}

@book{CF67,
	author = {Cassels, John William Scott and Fr{\"o}hlich, Albrecht},
	booktitle = {Proceedings of an instructional conference organized by the {L}ondon {M}athematical {S}ociety (a {NATO} {A}dvanced {S}tudy {I}nstitute) with the support of the {I}nternational {M}athematical {U}nion},
	mrclass = {00.04 (10.00)},
    pages = {xvii + 366},
	publisher = {Academic Press London; Thompson Book Co. Inc., Washington D.C.},
	title = {Algebraic number theory},
	year = {1967}}

@inproceedings{defrutosfernandez22,
	address = {Dagstuhl, Germany},
	author = {de Frutos-Fern{\'a}ndez, Mar{\'\i}a In{\'e}s},
	booktitle = {13th International Conference on Interactive Theorem Proving (ITP 2022)},
	doi = {10.4230/LIPIcs.ITP.2022.14},
	editor = {Andronick, June and de Moura, Leonardo},
	isbn = {978-3-95977-252-5},
	issn = {1868-8969},
	pages = {14:1--14:18},
	publisher = {Schloss Dagstuhl -- Leibniz-Zentrum f{\"u}r Informatik},
	series = {Leibniz International Proceedings in Informatics (LIPIcs)},
	title = {{Formalizing the Ring of Ad\`{e}les of a Global Field}},
	url = {https://drops.dagstuhl.de/opus/volltexte/2022/16723},
	urn = {urn:nbn:de:0030-drops-167232},
	volume = {237},
	year = {2022}}

@inproceedings{dFFN24,
	address = {New York, NY, USA},
	author = {de Frutos-Fern{\'a}ndez, Mar{\'\i}a In{\'e}s and {Nuccio Mortarino Majno di Capriglio}, Filippo A. E.},
	booktitle = {Proceedings of the 13th ACM SIGPLAN International Conference on Certified Programs and Proofs},
	doi = {10.1145/3636501.3636942},
	isbn = {9798400704888},
	numpages = {15},
	pages = {190--204},
	publisher = {Association for Computing Machinery},
	series = {CPP 2024},
	title = {A Formalization of Complete Discrete Valuation Rings and Local Fields},
	url = {https://doi.org/10.1145/3636501.3636942},
	year = {2024}}

@book{serre79,
	author = {Serre, Jean-Pierre},
	doi = {https://doi.org/10.1007/978-1-4757-5673-9},
	isbn = {9780387904245},
	location = {New York},
	publisher = {Springer New York},
	series = {Graduate Texts in Mathematics},
	title = {Local Fields},
	year = {1979}}

@inproceedings{mathlib,
	author = {{\ignorespaces The} {mathlib Community}},
	booktitle = {Proceedings of the 9th ACM SIGPLAN International Conference on Certified Programs and Proofs},
	collection = {POPL '20},
	doi = {10.1145/3372885.3373824},
	month = jan,
        year = {2020},
	title = {The {L}ean {M}athematical {L}ibrary},
publisher = {Association for Computing Machinery},
url = {https://doi.org/10.1145/3372885.3373824},
doi = {10.1145/3372885.3373824},
booktitle = {Proceedings of the 9th ACM SIGPLAN International Conference on Certified Programs and Proofs},
pages = {367–381},
numpages = {15},
location = {New Orleans, LA, USA},
series = {CPP 2020}
}

@article{ostrowski16,
author = {Alexander Ostrowski},
title = {{Über einige Lösungen der Funktionalgleichung $\psi(x)\cdot\psi(x)=\psi(xy)$}},
volume = {41},
journal = {Acta Mathematica},
number = {none},
publisher = {Institut Mittag-Leffler},
pages = {271 -- 284},
year = {1916},
doi = {10.1007/BF02422947},
URL = {https://doi.org/10.1007/BF02422947}
}

@inproceedings{tate50,
	author = {John T. Tate},
	booktitle = {Algebraic Number Theory},
	series = {Proc. Instructional Conf., Brighton, 1965},
	title = {Fourier analysis in number fields, and Hecke's zeta-functions},
	url = {https://api.semanticscholar.org/CorpusID:117898073},
	year = {1950}}

@InProceedings{vandoorn21,
  author =	{van Doorn, Floris},
  title =	{{Formalized Haar Measure}},
  booktitle =	{12th International Conference on Interactive Theorem Proving (ITP 2021)},
  pages =	{18:1--18:17},
  series =	{Leibniz International Proceedings in Informatics (LIPIcs)},
  ISBN =	{978-3-95977-188-7},
  ISSN =	{1868-8969},
  year =	{2021},
  volume =	{193},
  editor =	{Cohen, Liron and Kaliszyk, Cezary},
  publisher =	{Schloss Dagstuhl -- Leibniz-Zentrum f{\"u}r Informatik},
  address =	{Dagstuhl, Germany},
  URL =		{https://drops.dagstuhl.de/entities/document/10.4230/LIPIcs.ITP.2021.18},
  URN =		{urn:nbn:de:0030-drops-139139},
  doi =		{10.4230/LIPIcs.ITP.2021.18}
}

@book{weil1938,
	author = {Weil, A.},
	lccn = {ac39001341},
	publisher = {Hermann \& cie},
	series = {Actualit{\'e}s scientifiques et industrielles},
	title = {Sur les espaces {\`a} structure uniforme et sur la topologie g{\'e}n{\'e}rale},
	url = {https://books.google.co.uk/books?id=qVfvAAAAMAAJ},
	year = {1938}}
